\documentclass[a4paper,UKenglish,cleveref, autoref, thm-restate]{lipics-v2021}

\graphicspath{{./graphics/}}

\title{Process-Algebraic Models of Multi-Writer Multi-Reader Non-Atomic Registers}

\titlerunning{Process-Algebraic Models of MWMR Non-Atomic Registers}

\author{Myrthe Spronck}{Eindhoven University of Technology, The Netherlands}{m.s.c.spronck@student.tue.nl}{0000-0003-2909-7515}{}

\author{Bas Luttik}{Eindhoven University of Technology, The Netherlands}{s.p.luttik@tue.nl}{0000-0001-6710-8436}{}

\authorrunning{M.S.C. Spronck and B. Luttik} 

\Copyright{Myrthe Spronck and Bas Luttik} 

\ccsdesc[500]{Theory of computation~Shared memory algorithms}
\ccsdesc[500]{Theory of computation~Verification by model checking}

\keywords{mutual exclusion, model checking, non-atomic reads and writes, regular register.} 

\category{} 

\acknowledgements{We thank Rob van Glabbeek for insightful discussions on the topic of this paper.}

\nolinenumbers 

\hideLIPIcs  

\usepackage{graphicx}
\usepackage{color}
\usepackage{amsmath,amssymb}
\usepackage{calc,xspace,stmaryrd}
\usepackage{proc}
\usepackage{tikz}
\usetikzlibrary{arrows,shapes,automata,positioning,arrows.meta}
\tikzset{
        state/.style={
          circle,
         draw,
          minimum size=6mm,
    },
}
\usepackage{ulem}
\usepackage{todonotes}
\usepackage[noend]{algpseudocode}
\usepackage{algorithm}
\usepackage{rotating}

\usepackage{listings}

\lstdefinelanguage{mCRL2}
{
 keywords={act,var,cons,end,eqn,glob,init,val,whr,sort,map,pbes,proc,struct},
 keywords=[2]{true,false,delta,tau},
 keywords=[3]{Bool,Nat,Real,Pos,Int,Set,Bag,List,Int2Nat,Pos2Nat,Int2Pos,min,max
},
 keywords=[4]{hide,if,rename,sum,in,mu,nu,forall,exists,mod,allow,block,comm},
 keywords=[5]{nested,initial,state},
 numberstyle=\tiny\bfseries, 
 comment=[l]\%,
 commentstyle=\slshape,
 keywordstyle=[1]\bfseries,
 keywordstyle=[2]\itshape, 
 keywordstyle=[3]\itshape,
 keywordstyle=[4]\itshape,
 keywordstyle=[5]\bfseries\itshape,
 basicstyle=\ttfamily\scriptsize,
 flexiblecolumns=false
}
[keywords,comments]

\lstdefinelanguage{mCRL2-inline}
{
 keywords={act,var,cons,end,eqn,glob,init,val,whr,sort,map,pbes,proc,struct},
 keywords=[2]{true,false,delta,tau},
 keywords=[3]{Bool,Nat,Real,Pos,Int,Set,Bag,List,Int2Nat,Pos2Nat,Int2Pos,min,max
},
 keywords=[4]{hide,if,rename,sum,in,mu,nu,forall,exists,mod,allow,block,comm},
 keywords=[5]{nested,initial,state},
 numberstyle=\scriptsize\bfseries, 
 comment=[l]\%,
 commentstyle=\slshape,
 keywordstyle=[1]\bfseries,
 keywordstyle=[2]\itshape, 
 keywordstyle=[3]\itshape,
 keywordstyle=[4]\itshape,
 keywordstyle=[5]\bfseries\itshape,
 basicstyle=\ttfamily\small,
 flexiblecolumns=false
}
[keywords,comments]
\usepackage{booktabs}

\newcommand*{\defeq}{\stackrel{\text{def}}{\equiv}}

\newcommand{\lFor}[2]{
    \State\algorithmicfor\ {#1}\ \algorithmicdo\ {#2}%
  }
\newcommand{\lWhile}[2]{
    \State\algorithmicwhile\ {#1}\ \algorithmicdo\ {#2}%
  }

\newcommand{\co}{\ensuremath{\cdot}}
\newcommand{\then}{\ensuremath{\rightarrow}}

\newcommand{\tab}{\ensuremath{~~~~ ~~~~}}

\lstdefinelanguage{procTheory}
{
 numberstyle=\scriptsize\bfseries,
 basicstyle=\rmfamily\normalsize,
 escapeinside={@@},
 frame=single
}
\newcommand{\TID}{\ensuremath{\mathbb{T}}}
\newcommand{\Data}{\ensuremath{\mathbb{D}}}
\newcommand{\startread}[1][i]{\ensuremath{\mathit{sr_{#1}}}}
\newcommand{\finishread}[2][i]{\ensuremath{\mathit{fr_{#1}(#2)}}}
\newcommand{\rd}[2][i,j]{\ensuremath{\mathit{r_{#1}(#2)}}}
\newcommand{\startwrite}[2][i]{\ensuremath{\mathit{sw_{#1}(#2)}}}
\newcommand{\finishwrite}[1][i]{\ensuremath{\mathit{fw_{#1}}}}
\newcommand{\orderwrite}[1][i]{\ensuremath{\mathit{ow_{#1}}}}
\newcommand{\executeread}[1][i]{\ensuremath{\mathit{er_{#1}}}}
\newcommand{\executewrite}[1][i]{\ensuremath{\mathit{ew_{#1}}}}
\newcommand{\we}[2][i,j]{\ensuremath{\mathit{w_{#1}(#2)}}}

\newcommand{\ops}[1]{\ensuremath{\mathit{ops}(#1)}}
\newcommand{\reads}[1]{\ensuremath{\mathit{reads}(#1)}}
\newcommand{\readsno}[1]{\ensuremath{\mathit{reads}_\mathit{no}(#1)}}
\newcommand{\writes}[1]{\ensuremath{\mathit{writes}(#1)}}
\newcommand{\relwrites}[1]{\ensuremath{\textit{rel-writes}(#1)}}
\newcommand{\fixedwrites}[1]{\ensuremath{\textit{fix-writes}(#1)}}
\newcommand{\proj}[2]{\ensuremath{{#1}|{#2}}}
\newcommand{\sersym}[1][]{\ensuremath{\mathcal{S}_{#1}}}
\newcommand{\ser}[1][]{\ensuremath{\mathrel{\sersym[#1]}}}
\newcommand{\idlesym}{\ensuremath{\mathit{idle}}}
\newcommand{\idle}[1]{\ensuremath{\idlesym(#1)}}
\newcommand{\overlapsym}[1][i]{\ensuremath{\mathit{overlap}_{#1}}}
\newcommand{\overlap}[2][i]{\ensuremath{\overlapsym[#1](#2)}}
\newcommand{\posvalsym}[1][i]{\ensuremath{\mathit{pval}_{#1}}}
\newcommand{\posval}[2][i]{\ensuremath{\posvalsym[#1](#2)}}
\newcommand{\readerssym}{\ensuremath{\mathit{rdrs}}}
\newcommand{\readers}[1]{\ensuremath{\readerssym(#1)}}
\newcommand{\writerssym}{\ensuremath{\mathit{wrtrs}}}
\newcommand{\writers}[1]{\ensuremath{\writerssym(#1)}}
\newcommand{\pendingsym}{\ensuremath{\mathit{pndng}}}
\newcommand{\pending}[1]{\ensuremath{\pendingsym(#1)}}
\newcommand{\usrsym}[1][i]{\ensuremath{\mathit{usr}_{#1}}}
\newcommand{\usr}[2][i]{\ensuremath{\usrsym[#1](#2)}}
\newcommand{\ufrsym}[1][i]{\ensuremath{\mathit{ufr}_{#1}}}
\newcommand{\ufr}[2][i]{\ensuremath{\ufrsym[#1](#2)}}
\newcommand{\uswsym}[1][i]{\ensuremath{\mathit{usw}_{#1}}}
\newcommand{\usw}[2][i]{\ensuremath{\uswsym[#1](#2)}}
\newcommand{\ufwsym}[1][i]{\ensuremath{\mathit{ufw}_{#1}}}
\newcommand{\ufw}[2][i]{\ensuremath{\ufwsym[#1](#2)}}
\newcommand{\uowsym}[1][i]{\ensuremath{\mathit{uow}_{#1}}}
\newcommand{\uow}[2][i]{\ensuremath{\uowsym[#1](#2)}}
\newcommand{\uewsym}[1][i]{\ensuremath{\mathit{uew}_{#1}}}
\newcommand{\uew}[2][i]{\ensuremath{\uewsym[#1](#2)}}
\newcommand{\uersym}[1][i]{\ensuremath{\mathit{uer}_{#1}}}
\newcommand{\uer}[2][i]{\ensuremath{\uersym[#1](#2)}}
\newcommand{\nextvalsym}{\ensuremath{\mathit{next}}}
\newcommand{\nextval}[1]{\ensuremath{\nextvalsym(#1)}}
\newcommand{\wval}[2][i]{\ensuremath{\mathit{wval}_{#1}(#2)}}
\newcommand{\valssym}[1][i]{\ensuremath{\mathit{vals}_{#1}}}
\newcommand{\vals}[2][i]{\ensuremath{\valssym[#1](#2)}}

\newcommand{\Status}[1][m]{\ensuremath{\mathbb{S}_{#1}}}
\newcommand{\SStatus}{\ensuremath{\Status[s]}}
\newcommand{\RStatus}{\ensuremath{\Status[r]}}
\newcommand{\AStatus}{\ensuremath{\Status[a]}}
\newcommand{\Reg}[1][m]{R_{#1}}
\newcommand{\SReg}{\ensuremath{\Reg[s]}}
\newcommand{\RReg}{\ensuremath{\Reg[r]}}
\newcommand{\AReg}{\ensuremath{\Reg[a]}}
\newcommand{\undefsymb}{\ensuremath{\bot}}
\newcommand{\Traces}[1]{\ensuremath{\mathcal{T}_{#1}}}
\newcommand{\STraces}{\ensuremath{\Traces{s}}}
\newcommand{\RTraces}{\ensuremath{\Traces{r}}}
\newcommand{\ATraces}{\ensuremath{\Traces{a}}}

\newcommand{\step}[1]{\ensuremath{\overset{#1}{\longrightarrow}}}

\begin{document}
\lstset{
    language=mCRL2-inline,
    numbers=left}
\providecommand*{\lstnumberautorefname}{line} 

\renewcommand{\sectionautorefname}{Section}
\renewcommand{\subsectionautorefname}{Section}
\renewcommand{\subsubsectionautorefname}{Section}

\maketitle

\begin{abstract}
  We present process-algebraic models of multi-writer multi-reader safe, regular and atomic registers. We establish the relationship between our models and alternative versions presented in the literature. We use our models to formally analyse by model checking to what extent several well-known mutual exclusion algorithms are robust for relaxed atomicity requirements. Our analyses refute correctness claims made about some of these algorithms in the literature.
\end{abstract}

\section{Introduction}

  The mutual exclusion problem was first outlined by Dijkstra \cite{dijkstra65}. Given $n$ threads executing some code with a special section called the
  ``critical section'', the problem is to ensure that at any one time
  at most one of the threads is executing its critical section.
  Dijkstra explicitly states that
  communication between threads should be done through shared registers,
  and that reading from and writing to these registers should be 
  considered atomic operations; when two threads simultaneously interact
  with the register, be it through reading or writing, the register
  behaves as though these operations took place in some total order.
  
  Lamport argued that solutions to the mutual exclusion problem that assume atomicity of register operations do not fundamentally solve it \cite{Lamport86Mutex1}. After all, implementing atomic operations would require some form of mutual exclusion at a lower level.
  Many algorithms have been proposed that solve the mutual exclusion problem without requiring atomicity of register operations,
  most famously Lamport's own Bakery algorithm \cite{Lamport74}.
  
  Analysing distributed algorithms using non-atomic registers for communication between threads can be difficult, and correctness proofs are error-prone.
  Due to the vast number of execution paths of distributed algorithms, especially when
  overlapping register operations need to be taken into account, manual correctness proofs are likely to miss issues. One better uses computer tools (e.g., model checkers or theorem provers) to support correctness claims with a detailed and preferably exhaustive analysis. This introduces the need for formal models of non-atomic registers.
  
  Lamport proposed a general mathematical formalism for reasoning about the behaviour of concurrent systems that do not rely on the atomicity of operations, which he then uses to analyse the correctness of four solutions to the mutual exclusion problem not relying on atomicity \cite{Lamport86Mutex1, Lamport86Mutex2}. In \cite{Lamport86IPCbasic}, he studies in more detail the notion of single-writer multi-reader (SWMR) non-atomic register to implement communication between concurrent threads of computation; there, he distinguishes two variants, which he refers to as \emph{safe} and \emph{regular}. When a read operation to a SWMR safe register does not overlap with any write operations, then it will return the value stored in the register, but when it does overlap with a write operation then it may return a completely arbitrary value in the domain of the register. A SWMR regular register is a bit less erratic in the sense that a read operation overlapping with write operations will at least return any of the values actually being written.
  Raynal presented a straightforward generalisation of the notion of SWMR safe register to the multi-writer case \cite{Raynal13}. How the notion of SWMR regular register should be generalised to the multi-writer case, however, is less obvious. Shao et al.\ discuss four possibilities \cite{Shao11}. 

  The formalisms in \cite{Lamport86IPCbasic,Raynal13,Shao11} for studying the behaviour of non-atomic registers are not directly amenable for analysing the correctness of distributed algorithms by explicit-state model checking, e.g., using the mCRL2 toolset \cite{mCRL2toolset}. In fact, it is not clear whether the four variants of MWMR regular registers presented in \cite{Shao11} will lead to a finite-state model even if the number of readers and writers and the set of data values of the register are finite. 
  In \cite{lamportHyperbook}, Lamport demonstrates a method of modelling SWMR safe registers through repeatedly writing arbitrary values before settling on the desired value, but this approach does not generalise to multi-writer registers.
  The main contribution of this paper is to present process-algebraic models of multi-writer multi-reader safe, regular and also atomic registers that can be directly used in mCRL2 to analyse the correctness of distributed algorithms.

  We have used our process-algebraic models to analyse to what extent various mutual exclusion algorithms are robust for relaxed non-atomicity requirements. We find that Peterson's algorithm \cite{Peterson81} no longer guarantees mutual exclusion if the atomicity requirement is relaxed for the turn register. A variant of Peterson's algorithm presented in \cite{AttiyaWelch04} does guarantee mutual exclusion even if registers are only safe. The variant presented in \cite{Shao11}, however, does not guarantee mutual exclusion with regular registers, despite a claim that it does. We also find that some of the algorithms proposed in \cite{Szy88,Szy90} do not guarantee mutual exclusion for regular registers, which seems to contradict claims that they are immune to the problem of flickering bits during writes. When analysing Lamport's 3-bit algorithm \cite{Lamport86Mutex2} we discovered that its mutual exclusion guarantee crucially depends on how one of the more complex statements of the algorithm is implemented. Finally, we confirm that Aravind's BLRU algorithm \cite{aravind2010yet}, Dekker's algorithm \cite{alagarsamy2003some}, Dijkstra's algorithm \cite{dijkstra65} and Knuth's algorithm \cite{knuth1966additional} guarantee mutual exclusion even with safe registers.

  This paper is organised as follows. In \autoref{sec:registerModels} we present some basic definitions pertaining to SWMR registers, including formalisations of Lamport's notions of SWMR safe, regular and atomic registers. In \autoref{sec:MWMRdefs} we present and discuss our process-algebraic definitions of MWMR safe, regular and atomic registers, and establish formal relationships with their SWMR counterparts. In \autoref{sec:compare_shao} we compare our notion of MWMR regular register with the variants of MWMR regular registers proposed by \cite{Shao11}. In \autoref{sec:verifmutex} we report on our analyses of the various mutual exclusion algorithms. Finally, we present conclusions and some ideas for future work in \autoref{sec:conclusions}.

\section{Single-writer multi-reader registers}\label{sec:registerModels}

  The definitions presented in this section are adapted from \cite{Shao11} and \cite{Lamport86IPCalg}.

  We consider $n$ threads operating on a register with values in a finite set $\Data{}$ of register values; the initial value of the register will be denoted by $d_{\mathit{init}}$. Threads are identified by a natural number in the set $\TID=\{0,\dots,n-1\}$. A \emph{read operation} by thread $i\in\TID$ on the register, with \emph{return value} $d\in\Data$, is a sequence $\rd[i]{d}=\startread[i]\finishread[i]{d}$ consisting of an \emph{invocation} $\startread[i]$ (for ``thread $i$ starts to read''), and a matching \emph{response} $\finishread[i]{d}$ (for: ``the read by thread $i$ finishes with return value $d$''). A \emph{write} operation of thread $i$ on the register, with \emph{write value} $d$, is a sequence $\we[i]{d}=\startwrite[i]{d}\finishwrite[i]$ consisting of an \emph{invocation} $\startwrite[i]{d}$ (for: ``thread $i$ starts to write value $d$'') and a matching \emph{response} $\finishwrite[i]$ (for: ``the write by thread $i$ finishes''). An \emph{operation} of thread $i$ is either a read operation or a write operation of that thread.

  For every $i\in\TID$, let
    $A_i=\{\startread[i],\finishread[i]{d},\startwrite[i]{v},\finishwrite[i]\mid d\in\Data{}\}$,
  and let $A=\bigcup_{i\in\TID} A_i$. If $\sigma$ is a sequence of elements of $A$, then we denote by $\proj{\sigma}{i}$ the subsequence of $\sigma$ consisting of the elements in $A_i$.
  A \emph{schedule} on a register is a finite or infinite sequence $\sigma$ of elements of $A$ such that $\proj{\sigma}{i}$ consists of alternating invocations and matching responses, beginning with an invocation, and if $\proj{\sigma}{i}$ is finite, ending with a response. Note that, by these requirements and our definition of the notion of operation, $\proj{\sigma}{i}$ can then be obtained as the concatenation of read and write operations $o_0o_1o_2\dots$ executed by thread $i$.\footnote{The same operation may occur multiple times in $\proj{\sigma}{i}$. Henceforth, when we consider an operation in $\proj{\sigma}{i}$ we actually mean to refer to a specific occurrence in $\proj{\sigma}{i}$ of the operation. To disambiguate between two different occurrences of the same operation $o$ we could, e.g., annotate each occurrence of $o$ with its position in $\proj{\sigma}{i}$. We will not do so explicitly, because it will unnecessarily clutter the presentation. But the reader should keep in mind that, whenever we refer to an operation in a schedule $\sigma$ we actually mean to refer to a particular occurrence of that operation in $\proj{\sigma}{i}$.}  
  We shall denote by $\ops{\sigma,i}$ the set of all operations executed by thread $i$ (i.e., $\ops{\sigma,i}=\{o_0,o_1,o_2,\dots\}$) and by $\ops{\sigma}$ the set of all operations executed by any of the threads. It is technically convenient to include in $\ops{\sigma}$ a special write operation $w_\mathit{init}$that writes the initial value of the register. Then $\ops{\sigma}=\{w_{\mathit{init}}\}\cup\bigcup_{i\in\TID}\ops{\sigma,i}$. We also use $\reads{\sigma}$ and $\writes{\sigma}$ for the subsets of $\ops{\sigma}$ respectively consisting of the read operations and the write operations only.

  A schedule $\sigma$ induces a partial order on $\ops{\sigma}$: if $o,o'\in\ops{\sigma}$, then we write $o <_\sigma o'$ if, and only if, the response of $o$ precedes the invocation of $o'$ in $\sigma$. We stipulate that $w_{\mathit{init}}<o$ for all $o\in\ops{\sigma}\backslash\{w_{\mathit{init}}\}$.
  Let $r\in\ops{\sigma}$ be a read operation and let $w\in\ops{\sigma}$ be a write operation.
  We say that $w$ is \emph{fixed} for $r$ if $w <_{\sigma} r$; $\fixedwrites{\sigma,r}$ denotes the set of all writes that are fixed for $r$.
  We say that $w$ is \emph{relevant} for $r$ if $r\not<_{\sigma}w$; $\relwrites{\sigma,r}$ denotes the set of all writes in $\ops{\sigma}$ that are relevant for $r$. Note that, by the inclusion of $w_{\mathit{init}}$, the sets $\relwrites{\sigma,r}$ and $\fixedwrites{\sigma,r}$ are non-empty for all $r\in\reads{\sigma}$.
  We say that $r\in\reads{\sigma}$ \emph{can read from} $w\in\writes{\sigma}$ if $w$ is relevant for $r$ and there does not exist $w'\in\writes{\sigma}$ such that $w<_{\sigma}w'<_{\sigma}r$.
  An operation $o$ has \emph{overlapping writes} if there exists $w\in\writes{\sigma}$ such that $o\not<_{\sigma}w$ and $w\not<_{\sigma}o$.

  In \cite{Shao11}, a register model is defined as a set of schedules satisfying certain conditions. Restricting attention to single-writer multi-reader (SWMR) registers only, Lamport considers three register models: safe, regular and atomic \cite{Lamport86IPCalg}. We proceed to define Lamport's models by formulating conditions on \emph{single-writer} schedules, i.e., schedules in which all write operations are by one particular thread. If $\sigma$ is a single-writer schedule, then, since a write cannot have overlapping writes, every non-empty finite set $W$ of writes has a $<_{\sigma}$-maximum, i.e., an element $w\in W$ such that $w'<_{\sigma}w$ for all $w'\in W\setminus\{w'\}$. Since writes that are fixed for $r$ have their responses in the finite prefix of $\sigma$ preceding the invocation of $r$, we have that $\fixedwrites{\sigma,r}$ is finite for every $r$. Since $\fixedwrites{\sigma,r}$ is non-empty, it always has a $<_{\sigma}$-maximum.
  
  A SWMR register is \emph{safe} if a read that does not have overlapping writes returns the most recently written value. A read that does have overlapping writes may return any arbitrary value in the domain $\Data$ of the register.
  
  \begin{definition}\label{def:SWMR-safe}
      A single-writer schedule $\sigma$ is \emph{safe} if every read $r$ without overlapping writes returns the value written by the $<_{\sigma}$-maximum of the set of $\fixedwrites{\sigma,r}$.
  \end{definition}

  A SWMR register is \emph{regular} if it is safe, and a read that has overlapping writes returns the value of one of the overlapping writes or the most recently written value.

\begin{definition}\label{def:SWMR-regular}
    A single-writer schedule $\sigma$ is \emph{regular} if every read $r$ returns either the value written by the $<_{\sigma}$-maximum of the set $\fixedwrites{\sigma, r}$ or the value of an overlapping write.
\end{definition}

  A SWMR register is \emph{atomic} if all reads and writes behave as though they occur in some definite order. A \emph{serialisation} is a total order $\sersym{}$ on a subset $O$ of $\ops{\sigma}$ that is \emph{consistent} with $<_{\sigma}$ in the sense that for all $o,o'\in O$ we have that $o<_{\sigma}o'$ implies $o\ser{}o'$.
  A serialisation $(O,\sersym{})$ is \emph{legal} if every read operation returns the value of the most recent write operation according to $\sersym{}$, that is, whenever $r\in O$ is a read operation with return value $v$, then $v$ is the write value of $\sersym{}$-maximum of $\relwrites{\sigma,r}$.

\begin{definition}\label{def:SWMR-atomic}
    A single-writer schedule $\sigma$ is \emph{atomic} if $\ops{\sigma}$ has a legal serialisation.
\end{definition}

 \section{Multi-writer multi-reader registers}\label{sec:MWMRdefs}

 We now want to define multi-write multi-reader (MWMR) safe, regular and atomic registers.
 Since our goal is to verify the correctness of mutual exclusion algorithms by model checking, we prefer operational, process-algebraic definitions of register models over definitions in terms of schedules. We are going to define register models by giving recursive process definitions that, given the state of the register, admit certain interactions with the register, resulting in an update of the state of the register. 
  Which information needs to be maintained in the state of the register depends on the register model, but the state of register should at least reflect which operations are currently active. So, with each register model $m \in \{s,r,a\}$ we associate a set of states $\Status$, and we assume that the following functions are defined on $\Status$:
  \begin{equation} \label{eq:standardmappings}
    \begin{array}{l}
    \readerssym, \writerssym,\idlesym: \Status{} \rightarrow \mathcal{P}(\TID)\\
    \usrsym,\ufrsym,\ufwsym: \Status{}\rightarrow\Status{}\\
    \uswsym: \Data\times\Status{} \rightarrow \Status{}\enskip.
    \end{array}
  \end{equation}
  The mappings $\readerssym$ returns the set of all threads that are currently reading, i.e., $i\in\readers{s}$ if, and only if, thread $i$ has invoked a read operation but the matching response has not yet occurred. Similarly, $\writerssym$ returns the set of all threads that are currently writing, and $\idlesym$ returns the set of all threads that are currently not reading and not writing. The mappings $\usrsym$, $\ufrsym$, $\uswsym$ and $\ufwsym$ perform update operations on the state of the register, corresponding to whether the most recent interaction of the register was an invocation ($\usrsym$) or response ($\ufrsym$) of a read, or an invocation ($\uswsym$) or a response ($\ufwsym$) of a write. The update operation $\uswsym$ also takes  the write value into account.

  In the remainder of this section we shall first present our models of MWMR safe, regular and atomic registers, and then comment on the representation of these models in mCRL2.
  
\subsection{MWMR Safe Registers}\label{sec:safe-reg}

  Lamport's SWMR safe register model (see \autoref{def:SWMR-safe}) accounts for how reads and writes behave when they do not have overlapping writes, and how reads behave when they do have overlapping writes. To generalise Lamport's notion to MWMR registers, we need to define how writes behave when they have overlapping writes. We follow Raynal's approach and define that when a write has overlapping writes, then its effect is that some arbitrary value in $\mathbb{D}$ is written to the register \cite{Raynal13}.

  \begin{figure}
\begin{multline*}
    \SReg(d:\Data,s:\SStatus) = \\
    \sum_{i\in\TID}
    \left(\begin{array}{ll}
        & (i\in\idle{s}) \then
             \startread[i]{}\co\SReg(d,\usr[i]{s}) \\
      + & (i\in\idle{s}) \then
             \sum_{d'\in\Data} \startwrite[i]{d'}\co\SReg(d,\usw[i]{d',s})\\
      + & (i\in\readers{s}\wedge\neg\overlap[i]{s}) \then
             \finishread[i]{d}\co\SReg(d,\ufr[i]{s}) \\
      + & (i\in\readers{s}\wedge\overlap[i]{s}) \then
             \sum_{d'\in\Data} \finishread[i]{d'}\co\SReg(d,\ufr[i]{s})\\
      + & (i\in\writers{s}\wedge\neg\overlap[i]{s}) \then
             \finishwrite[i]\co\SReg(\nextval{s},\ufw[i]{s})\\
      + & (i\in\writers{s}\wedge\overlap[i]{s}) \then
             \sum_{d'\in\Data} \finishwrite[i]\co\SReg(d',\ufw[i]{s})
    \end{array}
    \right)
\end{multline*} \caption{Safe register model}\label{tab:procsafe}
\end{figure}

  Our process-algebraic definition of a MWMR safe register is shown in \autoref{tab:procsafe}. The equation defines the behaviour of processes $R_s(d,s)$; the parameter $d\in\Data$ reflects the current value of the register, and the parameter $s\in\SStatus$ reflects its current state.  For the behaviour of the safe register it must be determined for every read or write operation of a thread whether, during its interaction with the register, there was an overlapping write operation by some other thread. Therefore, in addition to the functions specified in \autoref{eq:standardmappings}, we presuppose on $\SStatus$ a predicate $\overlapsym[i]$ such that $\overlap{s}$ holds if during the interaction of thread $i$ with the register there was an overlapping write by another thread. At the response of a write that is not overlapping with other writes, the current value $d$ of the register needs to be replaced by the write value. Hence, whenever a write is invoked, the write value is stored in $s$ through $\usw{s}$; this value can be retrieved with the mapping $\nextvalsym:\SStatus\rightarrow\Data$ if the write had no overlapping writes. If there were overlapping writes, $\nextvalsym$ is undefined.
  The right-hand side of the equation in \autoref{tab:procsafe} specifies the behaviour of the register using standard process-algebraic operations: $\cdot$ denotes sequential composition, $+$ denotes non-deterministic choice, $\rightarrow$ denotes a conditional, and $\sum$ denotes choice quantification \cite{mCRL2language}.

  The definition in \autoref{tab:procsafe} induces a transition relations $\step{a}$ ($a\in A$)
  on the set of tuples $\langle d,s\rangle$ ($d\in\Data$, $s\in\SStatus$).
  For instance, if $i\in\readers{s}$ and $\neg\overlap[i]{s}$, then there is a transition
  \begin{equation*}
      \langle d,s\rangle \step{\finishread[i]{d}}\langle d,\ufr[i]{s}\rangle\enskip,
  \end{equation*}
  according to the third summand of the definition in \autoref{tab:procsafe}; and if $i\in\writers{s}$ and $\overlap[i]{s}$, then, for every $d'\in\Data$, there is a transition
  \begin{equation*}
      \langle d,s\rangle \step{\finishwrite[i]}\langle d', \usw{s}\rangle\enskip,
  \end{equation*}
  according to the last summand of the definition in \autoref{tab:procsafe}.

  We let $s_{\mathit{init}}$ denote the \emph{initial state} of the safe register, and we define $\idle{s_{\mathit{init}}}=\TID$, $\writers{s_{\mathit{init}}}=\readers{s_{\mathit{init}}}=\emptyset$,
  $\overlap[i]{s}$ is false, and $\nextval{s}=d_{\mathit{init}}$.
  Henceforth, we shall abbreviate $R_s(d_{\mathit{init}},s_{\mathit{init}})$ by $R_s$.
  A \emph{trace} of $R_s$ is a finite or infinite sequence $a_0a_1\cdots a_{n-1}a_n\cdots$ of elements of $A$ such that there exist $d_0,d_1,d_2, \dots,d_n,\ldots\in\Data$ and $s_0,s_1,s_2,\dots,s_n,\ldots\in\SStatus$ with $d_0=d_{\mathit{init}}$ and $s_0=s_{\mathit{init}}$ and
    $\langle d_0,s_0\rangle
      \step{a_0}
    \langle d_1,s_1\rangle
      \step{a_1}
      \cdots
      \step{a_{n-1}}
    \langle d_n,s_n\rangle
      \step{a_{n}}
      \cdots$.
  We denote by $\STraces$ the set of all traces of $R_s$.
  A trace $\alpha\in\STraces$ is \emph{complete} if, for all $i\in\TID$, either $\proj{\alpha}{i}$ is infinite or $\proj{\alpha}{i}$ ends with a response. A \emph{single-writer} trace is a trace in which all invocations and responses of write operations are by the same thread.

  We argue that there is a one-to-one correspondence between the single-writer safe schedules and the single-writer complete traces of $R_s$. First, note that schedules and complete traces adhere to exactly the same restrictions regarding the order in which invocations and responses of read and write operations can occur: the invocation of an operation by some thread can only occur when that same thread is not currently executing another operation, and a response to some thread for an operation can only occur if the last interaction of that thread was, indeed, an invocation of that same operation. Write values are not restricted in schedules, nor in complete traces. Moreover, in the single-writer case the value of the parameter $d$ of the process $R_s$ will always be the write value of write operation of which the execution finished last. Finally, note that both in schedules and in complete traces of $R_s$, if a read operation overlaps with a write operation, then it may return any value, and if it does not, then it will, indeed, return the value of the most recent write operation.
  
  \begin{proposition}
      Every single-writer safe schedule is a trace of $R_s$, and every complete single-writer trace of $R_s$ is a safe schedule.
  \end{proposition}
  
\subsection{MWMR regular registers}\label{sec:regular-reg}

According to Lamport's definition of SWMR regular registers (see \autoref{def:SWMR-regular}), a read either returns the write value of the $<_{\sigma}$-maximum of $\fixedwrites{\sigma, r}$ or the value written by one of its overlapping writes.
When writes may have overlapping writes, then $\fixedwrites{\sigma,r}$ may not have a $<_{\sigma}$-maximum. It is then necessary to determine, for every read $r$, which of the $<_{\sigma}$-maximal elements of $\fixedwrites{\sigma,r}$ should be taken into account when determining the return value of $r$, and to what extent different reads should agree on this choice.

Our considerations are as follows. First, we want our MWMR regular register model to coincide with Lamport's SWMR regular register model when there are no writes overlapping other writes, so that our analyses of algorithms that rely on SWMR regular registers are valid with respect to Lamport's model. Second, our model should be suitable for explicit-state model checking. This precludes any definition that requires keeping track of unbounded information pertaining to the history of the computation. To limit the amount of information that the model is required to remember, we let the register commit to a unique value when there are no active writes. In this respect, our model deviates from three of the four models considered in \cite{Shao11}; in \autoref{sec:compare_shao} we provide a more detailed comparison.

To be consistent with Lamport's SWMR regular registers, a read $r$ should be able return the value of any \emph{overlapping} write. To determine which of the elements of the \emph{fixed} writes is taken into account when determining the return value of $r$, our model non-deterministically inserts a special \emph{order} action $\orderwrite{}$ somewhere between the invocation and the response of every write of every thread $i\in\TID$.
One may think of the order action as marking the moment at which the write truly takes place. Note that this order action is purely for modelling purposes, we make no claims on the implementation of a regular register.
The write value associated with the most recent order action preceding the invocation of a read (or the initial value if no order actions have occurred yet) is taken into account as possible return value for that read. Thus, a serialisation of all writes is generated on-the-fly through the order actions: all read operations agree on the order of the writes.

\begin{figure}[h]
\begin{multline*}
    \RReg(d:\Data,s:\RStatus) = 
    \sum_{i\in\TID}
    \left(\begin{array}{ll}
        & (i\in\idle{s}) \then
             \startread[i]{}\co\RReg(d,\usr[i]{s}) \\
      + & (i\in\idle{s}) \then
             \sum_{d'\in\Data} \startwrite[i]{d'}\co\RReg(d,\usw[i]{d',s})\\
      + & (i\in\readers{s}) \then
             \sum_{d'\in\posval[i]{s}} \finishread[i]{d'}\co\RReg(d,\ufr[i]{s})\\
      + & (i\in\pending{s}) \then \orderwrite[i]\co\RReg(\wval[i]{s},\uow[i]{s})\\
      + & (i\in\writers{s}\wedge i\not\in\pending{s}) \then
             \finishwrite[i]\co\RReg(d,\ufw[i]{s})
    \end{array}
    \right)
\end{multline*}
\caption{Regular register model}\label{tab:procregular}
\end{figure}

Our process-algebraic definition of a MWMR regular register is given in \autoref{tab:procregular}.
Here, $\RStatus$ denotes the set of possible states of the MWMR regular register. 
The register keeps track of the readers, writers and idle threads, similar to the safe register. 
It additionally keeps track of the set $\pending{s}$ of threads that have invoked a write but for which the order action has not yet occurred. The update function $\uowsym : \RStatus \rightarrow \RStatus$
associated with the order action $\orderwrite$ removes thread $i$ from $\pending{s}$. For every thread $i \in \pending{s}$, $\wval[i]{s}$ is the write value of that write; it is used to correctly update the current value $d$ of the register when $\orderwrite[i]$ occurs. For every thread $i \in \readers{s}$, $\posval[i]{s}$ is the set of values that a read $r$ invoked by thread $i$ may return. That is, it consist of the values of all writes overlapping with $r$ (thus far) and the value of the write with the most recent $\orderwrite[j]$ before the invocation of $r$.

For $i\in\TID$, let $A_i^r = A_i\cup\{ \orderwrite[i]\}$, and let $A^r = \bigcup_{i \in \mathbb{T}}A_i^r$.
The process definition in \autoref{tab:procregular} induces transition relations $\step{a}$ ($a\in A^r$) on the set of tuples $\langle d,s\rangle$ ($d\in\Data$, $s\in \RStatus$).
As before $\idle{s_{\mathit{init}}} = \mathbb{T}$, $\readers{s_{\mathit{init}}} = \writers{s_{\mathit{init}}} = \emptyset$. We also have $\pending{s_{\mathit{init}}} = \emptyset$, and $\posval[i]{s_{\mathit{init}}} = \emptyset$ for all $i \in \mathbb{T}$. The initial values for $\wval[i]{s_{\mathit{init}}}$ do not matter, since $\wval[i]{s}$ only matters when $i \in \pending{s}$.
We use $R_r$ to abbreviate $R_r(d_{\mathit{init}},s_{\mathit{init}})$, and define a trace of $R_r$, also as before, as a finite or infinite sequence of elements of $A^r$ appearing as labels in a transition sequence starting at $\langle d_{\mathit{init}},s_{\mathit{init}}\rangle$. We denote by $\RTraces$ the set of all traces of $R_r$.

Compared to schedules, the traces of $R_r$ have extra  $\orderwrite$ actions. If $\alpha$ is a finite or infinite sequence of elements of $A^r$, then we denote by $\bar{\alpha}$ the sequence of elements of $A$ obtained from $\alpha$ by deleting all occurrences of $\orderwrite$ ($i\in\TID$). We can then formulate a correspondence between the single-writer traces of $R_r$ (i.e., the traces in which all invocations and responses of write operations are by the same thread) and single-writer regular schedules.

  If writes have no overlapping writes, then the most recent order action when a read $r$ is invoked either corresponds to the $<_{\sigma}$-maximum of $\fixedwrites{\sigma, r}$, or to a write that overlaps with $r$.
  In the first case, the set of possible values that can be returned by the read according to our model will coincide with the set of possible values that it can return according to \autoref{def:SWMR-regular}.
  In the latter case, our model allows a subset of the values possible according to \autoref{def:SWMR-regular} to be returned.
  Hence, a read in our model never returns a value that could not be returned according to Lamport's SWMR definition of regular registers.
  Moreover, if there is a trace of $R_r$ in which the order action $\orderwrite$ of a write that overlaps with $r$ occurs before the invocation of $r$, then there also exist a trace in which it occurs after the invocation of $r$.
  Thus, the set of traces described by our model includes all regular schedules according to \autoref{def:SWMR-regular} whenever there are no writes overlapping other writes.
    
\begin{proposition}
    For every single-writer regular schedule $\sigma$ there is a trace $\alpha$ of $R_r$ such that $\bar{\alpha}=\sigma$, and if $\alpha$ is a complete single-writer trace of $R_r$, then $\bar{\alpha}$ is a regular schedule.
\end{proposition}

\subsection{MWMR atomic registers}\label{sec:atomic-reg}

\autoref{def:SWMR-atomic}, formalising Lamport's notion of SWMR atomic register, straightforwardly generalises to MWMR registers by omitting the single-writer restriction on schedules. Our process-algebraic model should generate the legal serialisation of all operations on-the-fly. To this end, we introduce, for every thread $i$, \emph{execution} actions
$\executeread[i]$ and $\executewrite[i]$ to mark the exact moment at which an operation is treated as occurring. An operation's execution action must, of course, occur between its invocation and response. The value that is returned at the response of a read is the value that the register stored at the moment of that read's execution; the register's stored value is updated to a write's value at that write's execution.

The process-algebraic model of our MWMR atomic register is shown in \autoref{tab:procatomic}. The set of states of $\AReg$ is denoted by $\AStatus$. In addition to the standard update functions, there are  extra update functions $\uersym[i], \uewsym[i] : \AStatus \rightarrow \AStatus$ for the execution actions. The effect of applying $\uersym[i]$ on $s$ is to store the current value $d$ of the register as the value that should be returned at the response of the active read by thread $i$; this value can then be retrieved with $\vals{s}$, and $\vals{s}=\undefsymb$ until then. The effect of applying $\uewsym[i]$ is to update the current value $d$ of the register to the write value of the active write by thread $i$; this value can also be retrieved with $\vals{s}$, and $\vals{s}=\undefsymb{}$ thereafter.
Note that, by setting $\vals{s}$ to $\undefsymb$ before a read has been executed and after a write has been executed, we can use $\vals{s}$ in combination with $\readers{s}$ and $\writers{s}$ to determine whether the execution of an operation has taken place.

\begin{figure}[h]
\begin{multline*}
    \AReg(d:\Data,s:\AStatus) = \\
    \sum_{i\in\TID}
    \left(\begin{array}{ll}
        & (i\in\idle{s}) \then
             \startread[i]{}\co\AReg(d,\usr[i]{s}) \\
      + & (i\in\idle{s}) \then
             \sum_{d'\in\Data} \startwrite[i]{d'}\co\AReg(d,\usw[i]{d',s})\\
      + & (i\in\readers{s} \wedge \vals[i]{s} = \undefsymb) \then\executeread[i]\co\AReg(d,\uer[i]{s}) \\
      + & (i\in\writers{s}\wedge\vals[i]{s} \neq \undefsymb) \then \executewrite[i]\co\AReg(\vals[i]{s},\uew[i]{s}) \\
      + & (i\in\readers{s}\wedge\vals[i]{s} \neq \undefsymb)\then \finishread[i]{\vals[i]{s}}\co\AReg(d,\ufr[i]{s}) \\
      + & (i\in\writers{s}\wedge\vals[i]{s} = \undefsymb)\then\finishwrite[i]\co\AReg(d,\ufw[i]{s})
      \end{array}
    \right)
\end{multline*}
\caption{Atomic register model}\label{tab:procatomic}
\end{figure}

For $i\in\TID$, let $A_i^a = A \cup \{\executeread[i], \executewrite[i]\}$, and let $A^a = \bigcup_{i \in \mathbb{T}}A_i^a$.
The process definition in \autoref{tab:procatomic} induces transition relations $\step{a}$ ($a\in A^a$) on the set of tuples $\langle d,s\rangle$ ($d\in\Data$, $s\in \AStatus$).
As before $\idle{s_{\mathit{init}}}=\TID$ and $\readers{s_{\mathit{init}}} = \writers{s_{\mathit{init}}} = \emptyset$; the initial values for $\vals[i]{s_{\mathit{init}}}$ do not matter.
We use $R_a$ to abbreviate $R_a(d_{\mathit{init}},s_{\mathit{init}})$, and define a trace of $R_A$, also as before, as a finite or infinite sequence of elements of $A^A$ appearing as labels in a transition sequence starting at $\langle d_{\mathit{init}},s_{\mathit{init}}\rangle$. We denote by $\ATraces$ the set of all traces of $R_a$.

Compared to schedules, the traces of $\AReg$ have extra $\executeread[i]$ and $\executewrite[i]$ actions. If $\alpha$ is a finite or infinite sequence of elements of $A^{a}$, then we denote by $\bar{\alpha}$ the sequence obtained from $\alpha$ by deleting all occurrences of $\executeread[i]$ and $\executewrite[i]$ for $i \in \mathbb{T}$.
The correspondence between atomic schedules and complete traces of $\AReg$ follows straightforwardly. It suffices to prove that $\AReg$  admits exactly those traces $\alpha$ such that there exists a legal serialisation of $\bar{\alpha}$. 
To this end, note that the execute actions provide such a serialisation, and the definition of $\AReg$ has the responses of operations behave in accordance with this serialisation.
\begin{proposition}
    For every atomic schedule $\sigma$ there is a trace $\alpha$ of $\AReg$ such that $\bar{\alpha}=\sigma$, and if $\alpha$ is a complete trace of $\AReg$, then $\bar{\alpha}$ is an atomic schedule.
\end{proposition}

\subsection{mCRL2 implementation}\label{sec:mcrl2model}

The mCRL2 toolset \cite{mCRL2toolset} provides tools for model checking and equivalence checking.
Models are defined in the mCRL2 language \cite{mCRL2language}, which comprises a process-algebraic specification language and facilitates the algebraic specification of data types.
Properties defined in the modal $\mu$-calculus can be checked on those models.
One nice feature of mCRL2 is that when a property does not hold a counterexample can be generated.
For more information we refer to \cite{mCRL2language} as well as the toolset's website\footnote{\href{https://www.mcrl2.org}{\texttt{https://www.mcrl2.org}}}.

We have implemented the models presented in Figures~\ref{tab:procsafe}, \ref{tab:procregular} and \ref{tab:procatomic} in the mCRL2 language.
By adding processes that model the threads executing the desired algorithm in a manner compatible with the interface of the register models, we can verify the same algorithm easily under different atomicity assumptions.
An added benefit is that we can assume different levels of atomicity for different registers simultaneously, so that we pinpoint exactly to what extent the algorithm is robust for non-atomicity.
The model can be found as part of the examples delivered with the mCRL2 distribution\footnote{\url{https://github.com/mCRL2org/mCRL2/tree/master/examples/academic/non-atomic_registers} (\texttt{972629b})}.

The mCRL2 language has support for standard data types such as sets, bags and arrays (implemented as mappings) as well an algebraic specification facility to define new datatypes.
This allows us to model the register models staying close to the process-algebraic models presented in this paper.

\section{Alternative definitions of MWMR regular registers}\label{sec:compare_shao}

In \cite{Shao11} four definitions for MWMR regular registers are proposed. These are formulated as conditions on schedules.
We discuss how our definition of MWMR regular registers relates to these definitions.

The following definition captures the weakest condition on schedules presented in \cite{Shao11}.
\begin{definition}\label{def:shao-weak}
    A schedule $\sigma$ satisfies the \emph{weak} condition if, for every read operation $r$ in $\ops{\sigma}$, there exists a legal serialisation of $\writes{\sigma} \cup \{r\}$. 
\end{definition}

It follows straightforwardly from our MWMR regular register definition that any complete trace $\alpha\in\RTraces$ , when transformed into a schedule $\bar{\alpha}$ by deleting the order actions, satisfies \autoref{def:shao-weak}. As explained in \autoref{sec:regular-reg}, our model generates a serialisation of all writes. For every read $r$ by thread $i$, it returns either the value of the last write in this serialisation before $\startread[i]$, or the value of one of the writes overlapping this read.
In both cases, we may obtain a legal serialisation of $\writes{\bar{\alpha}} \cup \{r\}$ by taking the serialisation of writes associated with $\bar{\alpha}$ and inserting $r$ right after the write that it reads from. This is consistent with $<_{\sigma}$ because the serialisation of the writes is, and $r$ will only be placed after a write that either has its response before the invocation of $r$, or that $r$ overlaps with.

\begin{proposition}
    If $\alpha\in\RTraces$ is complete, then the schedule $\bar{\alpha}$ satisfies the weak condition.
\end{proposition}

In all our MWMR register definitions it is the case that when no writes are active on a register, it stores a unique value. It reduces the burden of storing elaborate information on the execution history of the register, as would be necessary with the definitions of \cite{Shao11}, and thus leads to a smaller statespace.
A consequence of our choice is that not all schedules satisfying the weak condition can be generated by our model.

\begin{figure}[htb]
     \centering
     \begin{subfigure}[b]{0.4\textwidth}
         \centering
         \includegraphics[width=\textwidth]{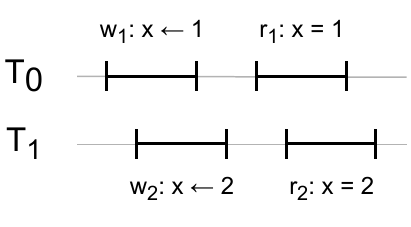}
         \caption{A schedule allowed by the weak, reads-from and no-inversion definitions but not by our regular register model.}
         \label{fig:weakcounter}
     \end{subfigure}
     \hfill
     \begin{subfigure}[b]{0.4\textwidth}
         \centering
         \includegraphics[width=\textwidth]{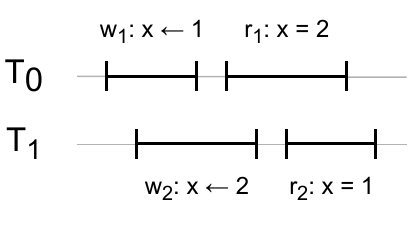}
         \caption{A schedule allowed by our regular register model but not by the write-order definition.}
         \label{fig:weakordercounter}
     \end{subfigure}
        \caption{Schedules demonstrating the differences between our regular register model and the definitions in \cite{Shao11}. We illustrate these schedules on a timeline, where an operation is drawn from its invocation to its response.}
        \label{fig:definitioncounters}
\end{figure}

\begin{example}
Consider the schedule depicted in \autoref{fig:weakcounter}. It is argued in \cite[Figure 6]{Shao11} that it satisfies the weak condition, but it cannot be generated by our regular register model $R_r$ because once $w_1$ and $w_2$ have ended, the register will have stored a unique value (either 1 or 2). Hence, the return values of $r_1$ and $r_2$ cannot be different.
Note that, for the same reason, the schedule cannot be generated by our safe register model $R_s$.
\end{example}

As illustrated in the preceding example, there exist schedules satisfying the weak condition that cannot be generated by our safe register model $R_s$. Conversely, it is easy to see that there exist complete traces generated by our safe register model $R_s$ (e.g., with overlapping writes resulting in a value that is not written by any of the writes) that do not satisfy the weak condition.

  The second condition in \cite{Shao11} associates with every read operation a serialisation and formulates a consistency requirement on these serialisations. If $r\in\reads{\sigma}$, then an \emph{$r$-serialisation} is a serialisation $\sersym[r]$ on $\relwrites{\sigma}\cup\{r\}$.\footnote{By considering serialisations of the relevant writes for $r$, instead of all writes, we deviate from \cite{Shao11}. Since a serialisation $\sersym[]$ on $\writes{\sigma}\cup\{r\}$ must be consistent with $<_{\sigma}$, we will have that $r \ser[] w$ for all $w\in\writes{\sigma}\setminus\relwrites{\sigma}$. It follows that the restriction of a serialisation $\sersym[]$ on $\writes{\sigma}\cup\{r\}$ to $\relwrites{\sigma}\cup\{r\}$ is an $r$-serialisation, and $\sersym[]$ is legal if, and only if, its restriction is.}
  
\begin{definition}\label{def:shao-write-order}
    A schedule $\sigma$ satisfies \emph{write-order} if for each read $r$ in $\ops{\sigma}$ there exists a legal serialisation $\ser_{r}$ of $\relwrites{\sigma} \cup \{r\}$ satisfying the following condition: for all reads $r_1$, $r_2$ in $\ops{\sigma}$, and for all writes $w_1, w_2 \in \relwrites{\sigma, r_1} \cap \relwrites{\sigma, r_2}$ it holds that $w_1 \ser_{r_1} w_2$ if and only if $w_1 \ser_{r_2} w_2$.
\end{definition}

\begin{proposition}\label{prop:weakorderreg}
    For every schedule $\sigma$ satisfying the write-order condition, there exists a trace $\alpha$ in $\RTraces$ such that $\bar{\alpha} = \sigma$.
\end{proposition}
We give a brief, informal description of how such a trace $\alpha$ can be constructed here; a more formal argument is presented in \autoref{app:proofMWRegWO}. The idea is that order actions can be inserted between the invocation and response of every write in $\sigma$, such that the return values of the reads match this placement of order actions.
Note that for reads that return the value of an overlapping write, this return value is possible according to \autoref{tab:procregular} regardless of how the order actions are placed. 
In our placement of order actions, we therefore only need to carefully consider reads that return the value of a write that is fixed for them.
According to \autoref{def:shao-write-order}, reads in $\sigma$ agree on the relative ordering of all writes that are relevant to them. Since $\fixedwrites{\sigma, r} \subseteq \relwrites{\sigma, r}$ for every read $r$, the reads also agree on the relative ordering of the fixed writes. 
We use this information to construct an ordering on all writes that is consistent both with $<_{\sigma}$ and with the return values of reads that read from writes that are fixed for them. 
Effectively, we find a single view on the relative order of all the write operations that is possible for every read in the schedule that returns the value of a fixed write.
Using this ordering, we can then place the order actions in the schedule $\sigma$ to create the trace $\alpha \in \RTraces$ such that $\bar{\alpha} = \sigma$.

Whilst every schedule satisfying \autoref{def:shao-write-order} corresponds to a trace of our model, not every schedule with a corresponding trace in our model is allowed by the write-order condition.
\begin{example}
Consider \autoref{fig:weakordercounter}.
This schedule is allowed by our model; $r_1$ can read 2 in $x$ because it overlaps with $w_2$ and it is possible for $r_2$ to read 1 if the order action of $w_2$ is done before the order action of $w_1$.
This schedule does not meet \autoref{def:shao-write-order} however; since both writes to $x$ are relevant for both reads, the two reads must agree on the respective order of the writes. For $r_2$ to read 1, it must be the case that $w_2 \ser_{r_2} w_1$. But since $w_1 < r_1$ according to the schedule, this means that $w_2 \ser_{r_1} w_1 \ser_{r_1} r_1$, so $r_1$ cannot read 2.
\end{example}

The third and fourth conditions on schedules proposed in \cite{Shao11} we refer to as \textit{reads-from} \cite[Definition 9]{Shao11} and \textit{no-inversion} \cite[Definition 10]{Shao11}, respectively. We do not recall these conditions here, and instead refer to \cite{Shao11} for more details.

Our notion of MWMR regular register is incomparable with the notions induced by the reads-from and no-inversion conditions on schedules.
First, as already indicated, every schedule that satisfies the write-order condition is also allowed by our model.
As it is proven in \cite{Shao11} that the write-order condition is incomparable with the reads-from and no-inversion conditions, this means our model admits schedules not admitted by these definitions.
To see that that not all schedules satisfying reads-from and no-inversion are admitted by our model, it suffices to observe that the schedule presented in \autoref{fig:weakcounter}, which is not admitted by our MWMR regular register model, satisfies  the reads-from and the no-inversion conditions. (See, e.g., \cite[Figure 8]{Shao11} and \cite[Figure 9]{Shao11}, which satisfy the reads-from and no-inversion conditions, respectively, and have the schedule in \autoref{fig:weakcounter} as prefix.)

\section{Verifying Mutual Exclusion Protocols}\label{sec:verifmutex}

\theoremstyle{remark}
\newtheorem{property}{Property}

We have used the register processes described in \autoref{sec:MWMRdefs} to analyse several well-known mutual exclusion algorithms. To this end, we have modelled the behaviour of the threads as prescribed by the algorithm also as processes, which interact with the register processes. That a thread is executing its non-critical section is represented in our model by the action $\mathit{noncrit}$, and that is executing its critical section is represented by the action $\mathit{crit}$; both actions are parameterised with the thread id.
We have checked the following two properties.
\begin{property}[Mutex]\label{prop:mutex}
    There is no state reachable from the initial state of the model in which there are two distinct threads $i$ and $j$ such that $\mathit{crit}(i)$ and $\mathit{crit}(j)$ are both enabled in this state.
\end{property}

\begin{property}[Reach]\label{prop:reach}
    For all threads $i$, always after an occurrence of a $\mathit{noncrit}(i)$ action it holds that, as long as a $\mathit{crit}(i)$ action has not happened, a state is reachable in which $\mathit{crit}(i)$ is enabled.
\end{property}

The Reach property is implied by starvation freedom, and so if it does not hold, then neither does starvation freedom.
We chose to analyse this property rather than starvation freedom itself because the presence of busy waiting loops in our models would require us to use fairness assumptions to dismiss spurious counterexamples.
The question of how to interpret fairness assumptions when dealing with non-atomic registers is outside of the scope of this paper.

The results of our verification are shown in \autoref{tab:verifres}.
When doing model checking, we have to instantiate a specific number of threads.
We have restricted our verification to three threads for all algorithms, except for Dekker, Attiya-Welch and Peterson, which are only defined for two threads. 

In this section, we discuss some of our most interesting findings.
For complete descriptions of counterexamples, as well as further discussion of our results we refer to \autoref{app:mutex}. All models are available through GitHub\footnote{\url{https://github.com/mCRL2org/mCRL2/tree/master/examples/academic/non-atomic_registers} (\texttt{972629b})}.

\begin{table}[t]
    \centering
\begin{tabular}{@{}l|ll|ll|ll@{}}
\toprule
 & \multicolumn{2}{c|}{Safe} & \multicolumn{2}{c|}{Regular} & \multicolumn{2}{c}{Atomic} \\
 & Mutex & Reach & Mutex & Reach & Mutex & Reach \\ \midrule
Aravind (BLRU) \cite[Figure 4]{aravind2010yet}  &  \checkmark & \checkmark & \checkmark & \checkmark  & \checkmark & \checkmark \\
Attiya-Welch \cite[Algorithm 12]{AttiyaWelch04}  & \checkmark & \checkmark & \checkmark & \checkmark & \checkmark & \checkmark \\
Attiya-Welch alternate \cite[Figure 19.1]{Shao11} & \checkmark & \texttimes & \checkmark & \texttimes & \checkmark & \checkmark\\
Dekker \cite[Figure 1]{alagarsamy2003some}& \checkmark & \checkmark & \checkmark & \checkmark & \checkmark & \checkmark \\
Dijkstra \cite{dijkstra65} & \checkmark  &  \checkmark & \checkmark & \checkmark & \checkmark & \checkmark  \\
Knuth \cite{knuth1966additional}& \checkmark  & \checkmark & \checkmark  & \checkmark & \checkmark & \checkmark \\
Lamport (3-bit) \cite[Figure 2]{Lamport86Mutex2}& \checkmark  & \checkmark  & \checkmark  & \checkmark& \checkmark  & \checkmark \\
Peterson \cite{Peterson81} & \texttimes & \checkmark & \texttimes & \checkmark & \checkmark & \checkmark \\
Szymanski (flag) \cite[Figure 2]{Szy88}& \texttimes & \texttimes& \texttimes & \checkmark & \checkmark  & \checkmark\\
Szymanski (flag with bits) & \texttimes & \checkmark& \texttimes & \checkmark & \texttimes & \checkmark\\
Szymanski (3-bit lin.\ wait) \cite[Figure 1]{Szy90}& \texttimes & \checkmark & \texttimes  & \checkmark & \texttimes & \checkmark \\ \bottomrule
\end{tabular}
    \caption{Results of verifying mutual exclusion algorithms.  }
    \label{tab:verifres}
\end{table}

\subsection{Peterson's Algorithm}\label{subsec:peterson-mutex}
Peterson's classic algorithm (see \autoref{alg:peterson}) was not designed to be correct under non-atomic register assumptions. 
An analysis of the mutual exclusion violation with safe registers still gives interesting insights into the algorithm and some of the unexpected behaviour of safe registers.

As expected, mCRL2 reports that mutual exclusion does not hold when using non-atomic registers.
We present a visualisation of the counterexample generated by mCRL2 for safe registers in \autoref{fig:peterson_counter_mutex_safe_timeline}.
There are two instances of overlapping operations.
First, since the two writes to $\mathit{turn}$, labelled $w_3$ and $w_4$ in \autoref{fig:peterson_counter_mutex_safe_timeline}, overlap, according to the safe register model the register can have any arbitrary value after they both have ended. 
In this counterexample, $\mathit{turn}$ has the value 1, which allows thread 0 to read the value 1 (the read labelled $r_4$) and enter the critical section.
Second, thread 1's read of $\mathit{turn}$ (labelled $r_2$) overlaps with thread 0's write (labelled $w_4$). 
The read can therefore return an arbitrary value, in this case the value 0, which allows thread 1 to enter the critical section.

\begin{figure}[b]
    \centering
    \includegraphics[width=0.8\textwidth]{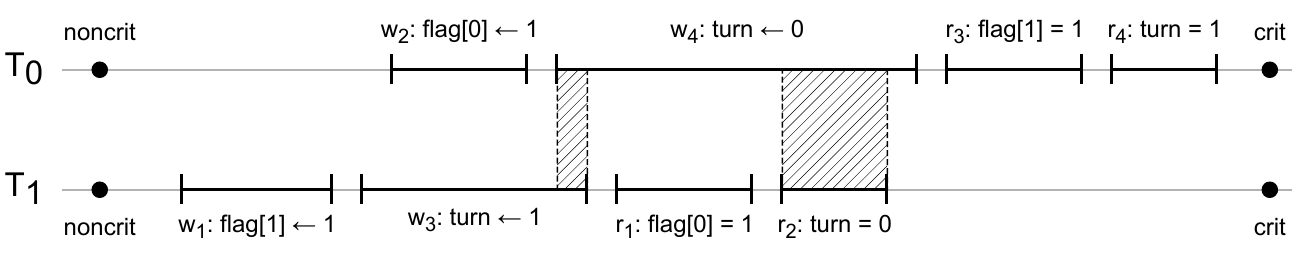}
    \caption{Counterexample generated by mCRL2 for mutual exclusion for Peterson's algorithm with safe registers, represented on a timeline}
    \label{fig:peterson_counter_mutex_safe_timeline}
\end{figure}

This counterexample shows only overlaps on the $\mathit{turn}$ register. 
We can initialise our model such that the $\mathit{turn}$ register is atomic, but both $\mathit{flag}$ registers behave as safe registers.
We find that mutual exclusion does hold then.
This confirms that overlapping operations on the $\mathit{turn}$ register are the sole cause of the mutual exclusion violation for Peterson's algorithm.
We discuss Peterson's algorithm with regular registers in \autoref{app:mutex}.

\noindent\begin{minipage}[t]{0.45\textwidth}
\begin{algorithm}[H]
\caption{Peterson's algorithm for two threads from \cite{Peterson81}. We use $i$ for the thread's own id and $j$ for the other thread's id.}\label{alg:peterson}
\begin{algorithmic}[1]
    \State{$\mathit{flag}[i]\gets 1$}
    \State{$\mathit{turn}\gets i$}
    \State{\textbf{await} $\mathit{flag}[j] = 0 \vee turn = j$}
    \State{\textbf{critical section}}
    \State{$\mathit{flag}[i]\gets 0$}
\end{algorithmic}
\end{algorithm}
\end{minipage}
\hfill
\begin{minipage}[t]{0.47\textwidth}
 \begin{algorithm}[H]
\caption{Szymanski's flag algorithm from \cite{Szy88}, $i$ is the thread's own id.}\label{alg:Szy-flag}
\begin{algorithmic}[1]
  \State{$\mathit{flag}[i]\gets 1$}
  \State{\textbf{await} $\forall j.\ \mathit{flag}[j]<3$}
  \State{$\mathit{flag}[i]\gets 3$}
  \If{$\exists j.\ \mathit{flag}[j]=1$}
    \State{$\mathit{flag}[i]\gets 2$}
    \State{\textbf{await} $\exists j.\ \mathit{flag}[j]=4$}
  \EndIf
  \State{$\mathit{flag}[i]\gets 4$}
  \State{\textbf{await} $\forall j < i.\ \mathit{flag}[j]<2$}
  \State{\textbf{critical section}}
  \State{\textbf{await} $\forall j>i.\ \mathit{flag}[j]<2 \vee
    \mathit{flag}[j]>3$}\label{alg:Szy-flag:exit}
  \State{$\mathit{flag}[i]\gets 0$}
\end{algorithmic}
\label{alg:szymanski-flag}
\end{algorithm}   
\end{minipage}

\subsection{Szymanski's Flag Algorithm}\label{subsec:szymanski-flag-mutex}
There are several variants of Szymanski's algorithm, which all seem to have been derived from the flag-based algorithm shown as \autoref{alg:szymanski-flag}. In \cite{Szy88}, Szymanski proposes this flag-based algorithm and claims that an implementation of it representing the flags using three bits is robust for flickering of bits (i.e., is correct for non-atomic registers).
As indicated in \autoref{tab:verifres}, we find that neither the integer nor the bits variant ensure mutual exclusion when using non-atomic registers. 
The full analysis of the bits version, as well as a variant of it known as the 3-bit linear wait algorithm \cite{Szy90} are presented in \autoref{app:mutex}.
Here, we only discuss the integer version of the flag algorithm, as the counterexample against Mutex that we have found illustrates the core issue shared by all mentioned variants of Szymanski's algorithm.

The pseudocode for the flag algorithm is shown in \autoref{alg:Szy-flag}. It is originally presented in \cite[Figure 2]{Szy88}, but note that
we have repaired an obvious typo: \cite[Figure 2]{Szy88} erroneously
has a conjunction instead of a disjunction in line~\ref{alg:Szy-flag:exit}.
All $\mathit{flag}$ registers are initialised at 0.

See \autoref{fig:szymanski_counter_mutex_regular_timeline} for a visualisation of the counterexample for mutual exclusion with two threads and regular registers that we found using the mCRL2 toolset. 
The first instance of a read overlapping with a write is irrelevant, reading $\mathit{flag}[1] = 1$ would also have been possible without overlap. 
The other two instances of overlap are of interest. 
Thread 0 is writing the value 3 to $\mathit{flag}[0]$ and thread 1 reads $\mathit{flag}[0]$ twice while this write is active.
The first time it reads the new value (3), while the second time it reads the old value (1). 
Lamport specifically highlights that such a sequence is possible when using regular registers \cite{Lamport86IPCalg}.
Since only single-writer registers are used and write-order reduces to Lamport's definition of regular registers when single-writers in that case \cite{Shao11}, this counterexample is also valid for write-order.

\begin{figure}[ht]
    \centering
    \includegraphics[width=\textwidth]{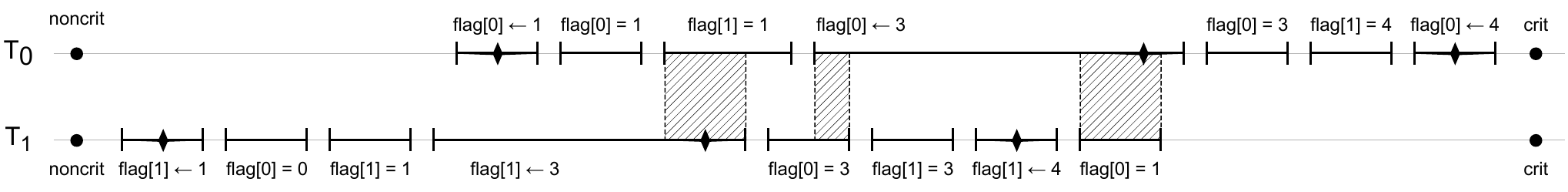}
    \caption{Mutex violation for Szymanski (flag) with regular registers and two threads, generated by mCRL2, on a timeline. The order-actions are drawn with lines during a write's execution.}
    \label{fig:szymanski_counter_mutex_regular_timeline}
\end{figure}

\subsection{Implementation Details}\label{subsec:implementation-mutex}
Our analyses have also revealed that seemingly minor implementation subtleties can make the difference between a correct and an incorrect algorithm.
A non-atomic register that is read multiple times in a row may return different values, even if no new writes to this register have started.
This means that when the value of a register needs to be checked several times in an algorithm, there is a difference between reading it once and subsequently checking a local copy of the value, or reading it again when needed.

For an example where this affects correctness, consider the Attiya-Welch algorithm. While the presentation in \cite[p.\ 77]{AttiyaWelch04} ensures reachability of the critical section with safe registers, the seemingly equivalent reformulation of this same algorithm in \cite{Shao11} does not.
The latter suggests that a thread needs to read a particular register twice as part of two different conditions that in the former are handled simultaneously.
In \cite{Shao11}, that presentation of the algorithm is claimed to be correct under all four of their MWMR regular register models; our counterexample shows that it is not.
A similar phenomenon occurs with Lamport's 3-bit algorithm, in which each thread $i$ has a bit $z_i$. As part of the algorithm, a computation is done on $z$ (the function assigning $z_i$ to $i$).
Lamport states that ``evaluating [$z$] at $j$ requires a read of the variable $z_j$.''
This may lead one to implement this algorithm by having threads re-read variables whenever needed. It turns out this implementation leads to a deadlock.
Locally saving all required $z$-values at the start of the computation and then only referencing this local copy during the computation solves this issue.
Consequently, these algorithms have a correct implementation, but they are also easily implemented incorrectly.
See the discussions of Attiya-Welch and Lamport in \autoref{app:mutex} for more details.

\subsection{Other Verifications}
There have been many mechanical verifications of mutual exclusion algorithms with atomic registers. For instance, in recent tutorials on the verification of distributed algorithms in mCRL2, verifications of Dekker's and Peterson's algorithms are presented \cite{groote2021tutorial,GKLVW19}. 
Several such verifications have also been done with the CADP toolset; see, e.g., \cite{mateescu2010study} for the results of verifying a large number of mutual exclusion algorithms, including Szymanski, Dekker and Peterson, with atomic registers.

To the best of our knowledge, we are the first to propose a systematic approach to mechanically verifying the correctness of mutual exclusion algorithms with respect to non-atomic registers, but there have been some mechanical verifications for specific algorithms.

Lamport himself modelled the Bakery algorithm in TLA+, representing the non-atomic writes as sequences of write actions of arbitrary length, where every action results in an arbitrary value being written, except for the last which writes the intended value \cite{lamportHyperbook}. This approach for modelling safe registers only works for SWMR registers; it does not work for MWMR registers. 
This approach for modelling safe SWMR registers, as well as a similar approach for modelling regular SWMR registers, is presented in \cite{anderson1988atomic}.
This approach is also used in several verifications done by Wim Hesselink, including of the Lycklama–Hadzilacos–Aravind algorithm in \cite{hesselink2015mutual} and the Bakery algorithm in \cite{hesselink2013mechanical}.

In \cite{cicirelli2016modelling}, several mutual exclusion algorithms are verified with atomic registers using timed automata in UPPAAL. Additionally, the Block $\&$ Woo algorithm is checked with bit flickering. Their model does not account for writes that overlap with other writes. Additionally, their model for the behaviour of safe registers is specific to the registers used in the algorithm.

Dekker's algorithm with safe registers is considered in \cite{buhr2016dekker}.
There it is demonstrated that Dekker's algorithm does not satisfy starvation freedom when safe registers are used, and a fixed version of the algorithm is presented.

Szymanski's flag algorithm with atomic registers is proven correct in \cite{manna1990exercise}. This paper demonstrates the importance of checking all threads in the ``forall'' and ``exists'' statements in the pseudocode in the same order every time. This is also how we model the algorithm.

There have been other verifications of Szymanski's algorithms \cite{manna1994step,szymanski1994automatic}, the former paper using the STeP tool. However, the exact pseudocode in those proofs differs from the pseudocode in \cite{Szy88} and \cite{Szy90}.

\section{Conclusions}\label{sec:conclusions}

We have presented process-algebraic models of safe, regular and atomic multi-writer multi-reader registers and used them to determine the robustness of various mutual exclusion algorithms for relaxed atomicity assumptions. Our analyses revealed issues with several of the algorithms discussed.

There are many more mutual exclusion algorithms that could be analysed in the same way as the ones shown in \autoref{sec:verifmutex}.
In \cite{Szy90}, Szymanski presents three other mutual exclusion algorithms. 
There also exist several variants of Szymanski's algorithm \cite{manna1994step,szymanski1994automatic}, all of which are similar to the 3-bit linear wait algorithm but differ in small ways.
In \cite{buhr2016dekker} it is shown that Dekker's algorithm does not ensure starvation freedom when safe registers are used and a modified version of the algorithm is presented which does satisfy this property. When we add verification of starvation freedom to our analysis, we can confirm their work.

We have only considered to what extent various algorithms guarantee mutual exclusion and whether the critical section is always reachable for every thread. Our next step will be to consider starvation freedom. Van Glabbeek proves that starvation freedom cannot hold for any mutual exclusion algorithm for which the correctness, on the one hand, relies on atomicity of memory interactions and, on the other hand, does not rely on assumptions regarding the relative speeds of threads \cite{glabbeek2021modelling}. A crucial presupposition for his argument is that a convincing verification hinges on not more than a component-based fairness assumption called justness. In \cite{BLW20} a method is proposed for verifying liveness properties under justness assumptions using the mCRL2 toolset. The method requires a classification of the roles of components in interactions. It should be investigated how to classify the roles of threads and registers in invocations and responses, and, in particular, how to deal with the \orderwrite{}, \executewrite{} and \executeread{} actions in the method.

  \bibliography{main}

\newpage
\appendix

  \section{Proof of \autoref{prop:weakorderreg}}\label{app:proofMWRegWO}

  This section presents our proof that every schedule $\sigma$ satisfying the write-order condition can be simulated by our model $R_r$.

  The return value of a read operation $r\in\ops{\sigma}$ is either the initial value of the register or the write value associate with some write operation $w\in\ops{\sigma}$ such that $r\not<_{\sigma}w$. Let $\sigma$ be a schedule that satisfies the write-order regular register condition; the \emph{reads-from mapping} $\rho$ for $\sigma$ is the mapping
  \begin{equation*}
      \rho: \reads{\sigma}\rightarrow \writes{\sigma}
  \end{equation*}
  that associates with every $r \in \reads{\sigma}$ its direct predecessor in $\sersym[r]$ (recall that we have included a special write operation $w_{\mathit{init}}$ in $\ops{\sigma}$ that precedes all other operations, so that every $r\in\reads{\sigma}$ indeed has a direct predecessor in $\sersym[r]$).
  
  \begin{proposition} \label{prop:readsfrom}
    Let $\sigma$ be a schedule satisfying the write-order regular register condition and let $\rho$ be the associated reads-from mapping. Then
      \begin{itemize}
      \item $\rho(r)\in\relwrites{r}$,
      \item there does not exist a write $w\in\writes{\sigma}$ such that $\rho(r)<_{\sigma}w <_{\sigma} r$, and
      \item the write value of $\rho(r)$ equals the return value of $r$.
  \end{itemize}
  \end{proposition}
    
  Our goal is now to show that every schedule satisfying the write-order regular register condition can be transformed into a trace of our regular register model by appropriately inserting $\orderwrite[i]$ actions.

  \begin{lemma} \label{lem:readsfromfinite}
     Let $w\in\writes{\sigma}$. If the set $\rho^{-1}(w)=\{r\in\reads{\sigma}\mid \rho(r)=w\}$ is infinite, then it has an infinite subset $R\subseteq\rho^{-1}(w)$ such that for all $r\in R$ and for all $w'\in \writes{\sigma}$ we have that $w' <_{\sigma} r$.
  \end{lemma}
  \begin{proof}
     Let $W=\{w'\in\writes{\sigma}\mid w\not<_\sigma w'\}$. Since the invocations of all $w'\in W$ must appear in $\sigma$ before the response of $w$, we have that $W$ is finite. At most finitely many reads can have their invocations appear before the last occurrence of a response of a write in $W$ and so for infinitely many reads $r\in \rho^{-1}(w)$ we have that $w' <_{\sigma} r$ for all $w'\in W$. Let $R$ be the set of all those reads, i.e.,
     \begin{equation*}
         R = \{ r\in \rho^{-1}(w) \mid \forall w'\in W. w'<_{\sigma} r\}\enskip.
     \end{equation*}
     It remains to argue that there cannot exist $w'\in\writes{\sigma}$ such that $w<_{\sigma}w'$. To this end, we derive a contradiction from the assumption that there does exist such $w'\in\writes{\sigma}$. Since only finitely many reads can have their invocations before the occurrence of the response of $w'$ in $\sigma$ (for the prefix of $\sigma$ including the response of $w'$ is finite), it follows that there exists $r\in R$ such that $w' <_{\sigma} r$. But then we have that $\rho(r) <_{\sigma} w' < r$, which contradicts the statements in Proposition~\ref{prop:readsfrom}.
  \end{proof}

  \begin{definition}
    We say that a read $r\in\reads{\sigma}$ is \emph{non-overlapping} if $\rho(r)<_{\sigma} r$. We denote the set of all non-overlapping reads in $\sigma$ by $\readsno{\sigma}$.
  \end{definition}
  
  \begin{lemma} \label{lem:ordering}
      Let $\sigma$ be a schedule that satisfies the write-order condition.
      Then there exists an enumeration $\vec{o}=o_0,o_1,o_2,\dots$ of $\readsno{\sigma}\cup\writes{\sigma}$ satisfying the following properties:
      \begin{enumerate}
          \item $o_0=w_{\mathit{init}}$;
          \item if $o_i<_{\sigma} o_j$, then $i<j$ for all relevant $i,j$;
          \item for every $r\in\readsno{\sigma}$ we have that
            $\rho(r)$ appears before $r$ in $\vec{o}$ and between $\rho(r)$ and $r$ there is no other write; and
          \item for all reads $r,r'\in\readsno{\sigma}$, if $r$ and $r'$ are distinct, $\rho(r)=\rho(r')$ and $r$ appears before $r'$ in $\vec{o}$, then the invocation of $r$ occurs before the invocation of $r'$ in $\sigma$.
      \end{enumerate}
  \end{lemma}
  \begin{proof}
      We define $\vec{o}$ in three steps: first, we define an enumeration of
      \begin{equation*}
          W = \bigcup_{r\in\reads{\sigma}}\relwrites{\sigma,r}\enskip;
      \end{equation*}
      then we extend this enumeration to an enumeration of $\writes{\sigma}$; and finally we suitably insert the elements of $\readsno{\sigma}$ in this enumeration. After defining $\vec{o}$ we shall establish that it satisfies the required properties. 
      
      Define the relation $\sersym{}$ on $W$ by
      \begin{equation*}
          \sersym{} = \bigcup_{r\in\reads{\sigma}}(\sersym[r]\cap(W\times W))\enskip.
      \end{equation*}
      
      To prove that $\sersym{}$ is irreflexive, it suffices to note that $\sersym[r]$ is irreflexive for all $r\in\reads{\sigma}$.
      
      To prove that $\sersym{}$ is transitive, let $w_1,w_2,w_3\in W$ and suppose that $w_1\ser{} w_2$ and $w_2\ser{} w_3$.
      From the definition of $\sersym$ and $w_1\ser w_2$ it follows that $w_2\neq w_\mathit{init}$; similarly, from $w_2\ser{}w_3$ it follows that $w_3\neq w_\mathit{init}$. Then there exist $r$ and $r'$ such that $w_1\ser[r]w_2$ and $w_2\ser[r']w_3$. If $r=r'$, then $w_1\ser{}w_3$ immediately follows by the transitivity of $\sersym[r]=\sersym[r']$. Otherwise, either the response of $r'$ occurs later in $\sigma$ than the response of $r$, or vice versa. In the first case, we have that $\relwrites{\sigma,r}\subseteq\relwrites{\sigma,r'}$ and therefore we get by the write-order regular register condition that $w_1\ser[r']w_2$; since $\sersym[r']$ is transitive, it follows that $w_1 \ser[r']w_3$ and hence $w_1 \ser{} w_3$. In the second case, we have that $\relwrites{\sigma,r'}\subseteq\relwrites{\sigma,r}$ and therefore we get by the write-order regular register condition that $w_2\ser[r]w_3$; since $\sersym[r]$ is transitive, it follows that $w_1\ser[r] w_3$ and hence $w_1\ser{} w_3$.

      To see that for all $w,w'\in W$ we either have that $w \ser{} w'$ or $w'\ser{} w$, note that there exists $r\in\reads{\sigma}$ such that $w$ and $w'$ are both relevant for $r$. Hence, either $w\ser[r]w'$ or $w'\ser[r]w$, so we have $w\ser w'$ or $w'\ser w$.
      
      Thus, we have now established that $\ser{}$ is a total order on $W$.

      Let $W'=\writes{\sigma}\setminus W$. If $W'\neq\emptyset$, then $W$ must be finite, for the writes in $W'$ are not relevant for any read in $\reads{\sigma}$ and hence their invocations all occur after the responses of all reads in $\sigma$. This means that there are finitely many reads in $\sigma$, and that $\sigma$ must have a finite prefix $\sigma'$ in which all responses of all reads in $\sigma$ occur. The writes in $W$ must have their invocations before the occurrence of the response of some read in $\sigma$, and so all occurrences of invocations of writes in $W$ must occur in the finite prefix $\sigma'$ of $\sigma$. This means that $W$ is finite.

      Now, let $\vec{w}=w_0,w_1,w_2,\dots$ be the enumeration of $\writes{\sigma}$ that starts with an enumeration of $W$ that is consistent with $\sersym{}$ (i.e., for all natural numbers $i,j$ we have that $i<j$ implies $w_i \ser{} w_j$), and that, if $W$ is finite,  is followed by an enumeration of $W'$ consistent with the order of the invocations of its elements in $\sigma$ (i.e., if $w_i,w_j\in W'$, then we have that $i<j$ implies that the invocation of $w_i$ occurs before the invocation of $w_j$ in $\sigma$).
      
      We proceed to extend $\vec{w}$ to an enumeration of all operations in $\readsno{\sigma}\cup\writes{\sigma}$ by inserting directly after each $w_i$ the elements of
      \begin{equation*}
          \rho^{-1}_\mathit{no}(w_i)=\{r\in\readsno{\sigma}\mid \rho(r)=w_i\}\enskip.
      \end{equation*}
      Let $\vec{r_i}=r_{i,0},r_{i,1},r_{i,2},\dots$ be the enumeration of $\rho^{-1}_{\mathit{no}}(w_i)$ that is consistent with the order of the invocations of the reads in $\sigma$ (i.e., for all relevant natural numbers $j,k$ we have that $j<k$ implies that the invocation of $r_{i,j}$ precedes the invocation of $r_{i,k}$ in $\sigma$). Note that, by Lemma~\ref{lem:readsfromfinite}, $\vec{r_i}$ can only be infinite if $w_i$ is the last element of $\vec{w}$. So we can now define an enumeration $\vec{o}$ of the operations in $\readsno{\sigma}\cup\writes{\sigma}$ as follows:
      \begin{equation*}
          \vec{o} = {o_0,o_1,o_2,\dots}
                  = {w_0,\vec{r_0},w_1,\vec{r_1},w_2,\vec{r_2},\dots} \enskip.
      \end{equation*}

      We proceed to argue that $\vec{o}$ satisfies the required properties.

      \begin{enumerate}
          \item It is immediate by our assumptions about $w_{\mathit{init}}$ that it is the least element with respect to $\sersym$. So it is the first element of the enumeration of $W$, and hence of $\vec{o}$.
          \item To prove that $o_i <_{\sigma} o_j$ implies $i<j$ for all relevant $i,j$, we assume that $o_i<_{\sigma} o_j$ and distinguish seven cases.
          \begin{itemize}
              \item If $o_i,o_j\in W$, then from $o_j\in W$ and $o_i<_{\sigma}o_j$ it follows that there exists $r\in\reads{\sigma}$ such that $o_i,o_j\in\relwrites{\sigma,r}$. Then $o_i\ser[r]o_j$, so $o_i \ser{} o_j$ and hence $o_i$ appears before $o_j$ in $\vec{o}$.
              \item If $o_i,o_j\in W'$, then from $o_i<_{\sigma}o_j$ it follows that the response of $o_i$, and hence also the invocation of $o_i$, appears before the invocation of $o_j$, and therefore $o_i$ appears before $o_j$ in $\vec{o}$.
              \item If $o_i\in W$ and $o_j\in W'$, then it is immediately clear from the definition of $\vec{o}$ that $i<j$.
              \item We show that $o_i\in W'$ and $o_j\in W$ is impossible. For suppose it is, then, since $o_j\in W$, there exists $r\in\reads{\sigma}$ such that $o_j$ is relevant for $r$, while, since $o_i\in W'$ we have that $r<_{\sigma} o_i$. From $o_i<_{\sigma}o_j$ it follows that $r <_{\sigma} o_j$, contradicting that $o_j$ is relevant for $r$.
              \item Consider $o_i,o_j\in\readsno{\sigma}$. From $o_i <_{\sigma} o_j$ it follows that the response, and hence the invocation, of $o_i$ appears before the invocation of $o_j$ in $\sigma$, so if $\rho(o_i)=\rho(o_j)$, then it immediately follows that $i<j$. So we proceed with the assumption that $\rho(o_i)\neq\rho(o_j)$. Since $o_i$ and $o_j$ are non-overlapping reads, we have that $\rho(o_i)<_{\sigma} o_i$ and $\rho(o_j)<_{\sigma} o_j$, and from $o_i<_{\sigma}o_j$ it, moreover, follows that $\rho(o_i)<_{\sigma} o_j$. Since both $\rho(o_i)$ and $\rho(o_j)$ are relevant for $o_j$, we have that $\rho(o_i)$ and $\rho(o_j)$ are ordered by the $o_j$-serialisation $\ser[o_j]$. Moreover, we must have $\rho(o_i)\ser[r_j]o_j$ and $\rho(o_j)\ser[r_j] o_j$ and since $\rho(o_j)$ must be the direct $\ser[r_j]$-predecessor of $o_j$ it follows that $\rho(o_i)\ser[r_j]\rho(o_j)$. Hence, $\rho(o_i)$ appears before $\rho(o_j)$ in $\vec{o}$ and therefore $i<j$.
              \item If $o_i\in \writes{\sigma}$ and $o_j\in \readsno{\sigma}$, then both $o_i$ and $\rho(o_j)$ are relevant for $o_j$. Since $\sersym[o_j]$ must be consistent with $<_{\sigma}$ we have that $o_i\ser[o_j] o_j$. Furthermore, since $\rho(o_j)$ is defined as $o_j$'s direct predecessor according to $\ser[r_j]$, it follows that either $o_i\ser[o_j]\rho(o_j)$ or $o_i=\rho(o_j)$. In both cases, we find that $i<j$.
              \item Suppose that $o_i\in\readsno{\sigma}$ and $o_j\in\writes{\sigma}$. If $o_j\in W'$, then $\rho(o_i)<_{\sigma} o_i <_{\sigma} o_j$, so according to the definition of $\vec{o}$, $o_i$ will appear in $\vec{o}$ between $\rho(o_i)$ and $o_j$, from which it follows that $i<j$.
              If $o_j\in W$, then there exists a read $r$ such that $o_j\in\relwrites{\sigma,r}$. Since $\rho(o_i)<_{\sigma}o_i<_{\sigma}o_j$, it follows that $\rho(o_i)$ is also relevant for $r$. We must have $\rho(o_i)\ser[r]o_j$, and, according to the definition of $\vec{o}$, $o_i$ will appear between $\rho(o_i)$ and $o_j$, so $i<j$.
          \end{itemize}
          \item Let $r\in \readsno{\sigma}$. Then $r\in \rho^{-1}_{\mathit{no}}(\rho(r))$, so $r$ is an element of the sequence of writes directly following $\rho(r)$. It follows that $\rho(r)$ appears before $r$ in $\vec{o}$ and between $\rho(r)$ and $r$ there is no other write.
          \item Let $r,r'\in\readsno{\sigma}$ be distinct, suppose that $\rho(r)=\rho(r')$ and $r$ appears before $r'$ in $\vec{o}$. Then we have $r,r'\in\rho^{-1}_{\mathit{no}}(\rho(r))$, so $r$ and $r'$ both appear in the sequence of reads that directly follows $\rho(r)$ in $\vec{o}$. The reads in that sequence are ordered in accordance with the order of their invocations. It follows that the invocation of $r$ occurs before the invocation of $r'$ in $\sigma$.
          \qedhere
      \end{enumerate}
  \end{proof}

  The enumeration delivered by \autoref{lem:ordering} allows us to define a procedure that transforms a schedule $\sigma$ into a trace $\alpha\in\RTraces$ as follows. Simultaneously iterate through the enumeration and the schedule. If the current operation in the enumeration is a read, then we move forward in the schedule until after the invocation of $r$ and move to the next operation in the enumeration. If the current operation in the enumeration is a write, then there are two cases: if the invocation of the write has already occurred, then we insert the associated order action in the schedule, move past the order action, and move to the next operation; if the invocation of the write has not yet occurred, then we first move in the schedule until right after the invocation of the write, insert the associated order action, move past the order action, and move to the next operation. This establishes \autoref{prop:weakorderreg}.

\section{Mutual Exclusion Counterexamples}\label{app:mutex}

In this appendix we give in-depth discussions and example traces for our more interesting results presented in \autoref{sec:verifmutex}.
For all the models used as well as the exact pseudocode we modelled to arrive at the conclusions in \autoref{tab:verifres}, we point to the examples included with the mCRL2 distribution. These can be found at the same link as the register model itself.

In the counterexample discussions, we at times refer to the ``entry protocol'' and the ``exit protocol'' of an algorithm.
The former is the part of the algorithm before the critical section is entered, the latter is the part after the critical section is released.

\subsection{Properties}
In mCRL2 properties must be encoded in the modal $\mu$-calculus.
We encoded the mutual exclusion property as
\begin{gather*}
\forall_{i, j \in \mathbb{T}}. (i \neq j) \Rightarrow (\\
\tab [\mathit{true}^\star]\neg(\langle\mathit{crit}(i)\rangle \mathit{true} \land \langle\mathit{crit}(j)\rangle \mathit{true}))
\end{gather*}
and the reachability of the critical section property as
\begin{gather*}
\forall_{i \in \mathbb{T}}. \neg (\langle\mathit{true}^\star \cdot \mathit{noncrit}(i)\rangle \nu X.\langle\overline{\mathit{crit}(i)}\rangle X)
\end{gather*}

\subsection{Peterson}\label{app:peterson-mutex}
In \autoref{subsec:peterson-mutex} we demonstrate a counterexample for Peterson's algorithm using safe registers.
Here, we note that the same counterexample is also valid for our regular register model:
$r_2$ can still return a 0 because it overlaps with $w_4$; and by placing the $\mathit{order\_write}$ of $w_4$ before the $\mathit{order\_write}$ of $w_3$ we can also have thread 0 read a 1 at $r_4$.
However, this counterexample is not valid for the write-order definition from \cite{Shao11}. Since both writes to $\mathit{turn}$ are relevant for both the reads from $\mathit{turn}$, the two threads must agree on their respective order. For thread 0 to read a 1 in $\mathit{turn}$, it must be the case that $w_4 < w_3$. And since $w_3$ ends before $r_2$ starts, $r_2$ cannot possibly read a 0.
This is effectively the same situation as in \autoref{fig:weakordercounter}.
We cannot conclude from this that Peterson's algorithm is correct under write-order, only that this specific counterexample is not valid under those assumptions.

\subsection{Attiya-Welch}\label{app:attiya-welch-mutex}
The algorithm we refer to as the Attiya-Welch algorithm is presented in both \cite{AttiyaWelch04} and \cite{Shao11} as being Peterson's algorithm from \cite{Peterson81}.
While the algorithm indeed has some similarities to Peterson's it is not identical and in particular behaves different when using non-atomic registers. Hence why we refer to it as the Attiya-Welch algorithm, rather than a version of Peterson's.

As shown in the table, the Attiya-Welch algorithm does ensure mutual exclusion when non-atomic registers are used.
Of interest is that while the original presentation of the algorithm from \cite{AttiyaWelch04} also ensures reachability of the critical section when using non-atomic registers, the version of the algorithm presented in \cite{Shao11} does not. 
The two algorithms are shown in \autoref{alg:attiya-welch} and \autoref{alg:attiya-welch-var} respectively.

\begin{algorithm}
 \caption{Attiya-Welch algorithm for two threads as presented in \cite{AttiyaWelch04}. Once again $i$ is a thread's own id, $j$ the id of the other.}\label{alg:attiya-welch}
\begin{algorithmic}[1]
        \State{$\mathit{flag}[i] \gets 0$}\label{attiya-welch-entry2}
        \State{\textbf{await} $\mathit{flag}[j] = 0 \lor \mathit{turn} = j$}
        \State{$\mathit{flag}[i] \gets 1$}
        \If{$\mathit{turn} = i$}\label{attiya-welch-if}
            \If{$\mathit{flag}[j] = 1$}
                \State{\textbf{goto} line \ref{attiya-welch-entry2}}\label{attiya-welch-goto}
            \EndIf
        \Else
            \State{\textbf{await} $\mathit{flag}[j] = 0$}\label{attiya-welch-wait}
        \EndIf
        \State{\textbf{critical section}}
        \State{$\mathit{turn} \gets i$}
        \State{$\mathit{flag}[i] \gets 0$}
\end{algorithmic}
\end{algorithm}

\begin{algorithm}
\caption{Attiya-Welch algorithm for two threads as presented in \cite{Shao11} Once again $i$ is a thread's own id, $j$ the id of the other.}\label{alg:attiya-welch-var}
\begin{algorithmic}[1]
    \Repeat
    \State{$\mathit{flag}[i]\gets 0$}
    \State{\textbf{await} $\mathit{flag}[j] = 0 \vee turn = j$}
    \State{$\mathit{flag}[i]\gets 1$}
    \Until{$\mathit{turn} = j \vee \mathit{flag}[j] = 0$}\label{attiya-welch-var-until}
    \If{$\mathit{turn} = j$}\label{attiya-welch-var-check}
    \State{\textbf{await} $\mathit{flag}[j] = 0$}
    \EndIf
    \State{\textbf{critical section}}\label{attiya-welch-var-crit}
    \State{$\mathit{turn}\gets i$}
    \State{$\mathit{flag}[i]\gets 0$}
\end{algorithmic}
\end{algorithm}

When using regular registers and two threads, we get the following counterexample for reachability of the critical section on the variant presentation:
\begin{enumerate}
    \item Thread 1 gains uncontested access to the critical section. Because thread 0 is not competing, thread 1 can reach line \ref{attiya-welch-var-crit} without issue.
    \item Thread 1 starts its exit protocol with starting $\mathit{turn} \gets 1$. This operation does not have its order-action yet. At this point, a read of $\mathit{turn}$ can read both a 0 and a 1: 0 is the initial value and no $\mathit{order\_write}$ actions have occurred yet; 1 can be read because of overlap.
    \item Thread 0 starts the competition. Whenever it reads $\mathit{flag}[1]$ it sees a 1, but it can read whatever value it needs from $\mathit{turn}$ to get through the entry protocol. It escapes the first await-loop by reading $\mathit{turn} = 1$; it escapes the repeat-until loop by again reading $\mathit{turn} = 1$; it avoids the second await-loop entirely by reading $\mathit{turn} = 0$ on line \ref{attiya-welch-var-check}.
    \item The order-action for thread 1's write takes place, but the write is not finished yet.
    \item Thread 0 enters the critical section.
    \item In the exit protocol, thread 0 writes $\mathit{turn} \gets 0$. The order- and finish-actions take place immediately. At this point, the most recent order-action on $\mathit{turn}$ has the value 0, and the write by thread 1 is still active. So once again, reads of $\mathit{turn}$ can read either a 0 or a 1.
    \item Thread 0 finishes the exit protocol with $\mathit{flag}[0] \gets 0$. It re-enters the competition. Just like before, even though $\mathit{flag}[1] = 1$, thread 0 can get through most of the entry protocol by reading the required values for $\mathit{turn}$. It gets until right after line \ref{attiya-welch-var-until}, having just escaped the repeat-until loop. Unlike the previous execution of the entry protocol by thread 0, it reads $\mathit{turn} = 1$ on line \ref{attiya-welch-var-check}, as a result it starts awaiting $\mathit{flag}[1]$ to be 0.
    \item Thread 1 stops writing to $\mathit{turn}$. At this point, any read of $\mathit{turn}$ will give the value of the most recent order-action, which was 0.
    \item Thread 1 end the exit protocol by setting $\mathit{flag}[1]$ to 0. Thread 0 does not read $\mathit{flag}[1]$ yet.
    \item Thread 1 re-enters the competition. Since $\mathit{turn}$ is now 0, it can get through most of the entry protocol. It cannot get through the entry protocol entirely however because on line \ref{attiya-welch-var-check}, it once reads $\mathit{turn} = 0$. It then has to start awaiting $\mathit{flag}[0]$ to be equal to 0.
    \item At this point, both threads have their respective $\mathit{flag}$s set to 1, and both are waiting for the other's $\mathit{flag}$ to be 0. This is a deadlock; neither thread will ever be able to reach the critical section again.
\end{enumerate}

This same counterexample (ignoring the ordering of writes) also works for the safe registers.
This counterexample does not work for the write-order definition from \cite{Shao11} (it is once again similar to \autoref{fig:weakordercounter}), but it does hold for the weak definition.
This contradicts the claim in \cite{Shao11} that this algorithm ensures starvation freedom under weak -- if one thread can never reach the critical section again after having done a $\mathit{noncrit}$ action, then starvation freedom cannot hold.

As stated earlier, this counterexample is only present in the presentation from \cite{Shao11}, it is not present in the pseudocode given in \cite{AttiyaWelch04}.
At first glance, the algorithms seem to be equivalent: 
the goto-statement on line \ref{attiya-welch-goto}  of \autoref{alg:attiya-welch} is removed, and instead that part of the code is turned into a logically equivalent repeat-until loop.
There is only a minor implementation difference: where
in the original presentation $\mathit{turn}$ is read only once to determine whether the loop should be taken and whether the thread needs to wait for the other to lower its flag; 
in the variant presentation, $\mathit{turn}$ is read twice.

As shown in the counterexample, the deadlock requires an overlapping write on $\mathit{turn}$ which can only occur when both threads are in the exit protocol simultaneously. 
Since mutual exclusion is guaranteed, it must be the case that one thread writing to $\mathit{turn}$ is what allows the other thread to reach the exit protocol.
In \autoref{alg:attiya-welch} this is not possible: on line \ref{attiya-welch-if}, if $\mathit{turn} = i$ is read then, since $\mathit{flag}[j] = 1$, the other thread will be forced back to line \ref{attiya-welch-entry2}.
If on the other hand $\mathit{turn} = j$ is read, then the thread gets stuck in the waiting loop on line \ref{attiya-welch-wait}.
In \autoref{alg:attiya-welch-var} it is possible: by reading $\mathit{turn} = j$ on line \ref{attiya-welch-var-until} and $\mathit{turn} = i$ on line \ref{attiya-welch-var-check}. 
This is mentioned in part 3 in the counterexample.

\subsection{Lamport (3-bit)}\label{app:lamport-mutex}
Lamport highlighted the (theoretical) importance of mutual exclusion
algorithms that are resistant to safe registers \cite{Lamport86Mutex1}.
He also proposed the first algorithm specifically designed
to ensure mutual exclusion under safe registers:
the Bakery algorithm \cite{Lamport74}.
But the Bakery algorithm requires unbounded memory registers
and is therefore in many ways impractical and at the very least hard
to analyse.
In \cite{Lamport86Mutex2}, Lamport proposes
4 different solutions for this problem, each offering a different
trade-off between using a small number of communication variables
and satisfying stronger properties. We focus on the second algorithm,
which requires only three SWMR bits per thread and still ensures both
mutual exclusion and starvation freedom.
The algorithm can be seen in \autoref{alg:lamport3bit}.
We first introduce the variables used in the algorithm and the required definitions for understanding it.

This algorithm is for an arbitrary number of threads. We use id's $0$ to $N -1$ when there are $N$ threads.
The $j, y$ and $f$ variables are private variables in the range $0$ to $N - 1$.
The $x_i, y_i$ and $z_i$ registers are all Boolean variables initially set to 0/false.

Lamport's Three Bit Algorithm makes extensive use of cycles. A cycle,
as defined in \cite{Lamport86Mutex2}, is an object of the form
$\langle i_0, .., i_m \rangle$ of distinct elements. Two cycles
are the same if they contain the same elements in the same order
except for a cyclic permutation. The first element of a cycle is
its smallest element, so we take as the representative of a cycle
a list where the smallest element is at index 0. An ordered cycle
has all elements in order from smallest to largest, possibly only
after cyclic permutation. Since our representation of a cycle
is a list with the smallest element at the first position, 
an ordered cycle can be represented with a sorted list.

The operation $\mathrm{ORD}~S$ takes a set $S$ and returns the
ordered cycle containing exactly the elements from $S$.

In the algorithm, the Boolean function $CG(v, \gamma, i_j)$ is used.
Here, $v$ is a Boolean function mapping each element in the cycle
$\gamma$ to either true or false, and $i_j$ is an element from
$\gamma$.
\begin{align*}
   CG(v, \gamma, i_j) \defeq&~ v(i_j) \equiv CGV(v, \gamma, i_j) \\
   CGV(v, \gamma, i_j) \defeq&\begin{cases}
        \neg v(i_{j -1}) &\text{if $j > 0$}\\
        v(i_m) &\text{if $j = 0$}
        \end{cases}
\end{align*}

The phrase ``$i\gets j$\ \textbf{cyclically to} $k$'' means that
the iteration starts with $i = j$, then $j$ gets incremented by 1,
modulo the length of the cycle. This continues until $j = k$, at
which point the iteration stops without executing the loop with $j = k$.
$\oplus$ is used for addition modulo the length of the cycle.

\begin{algorithm}
\caption{Lamport's Three Bit algorithm}
\label{alg:lamport3bit}
\begin{algorithmic}[1]
  \State{$y_i \gets 1$} 
  \State{$x_i \gets 1$} \label{l1}
  \State{$\gamma \gets \mathrm{ORD}\{i\mid y_i = 1\}$} \label{l2}
  \State{$f\gets \mathrm{minimum}\{j\in\gamma \mid
    \mathit{CG}(z,\gamma,j)=1\}$}\label{lamport-3bit-functioncall}
  \For{$j\gets f$\ \textbf{cyclically to} $i$}
  \If{$y_j=1$}
    \If{$x_i=1$} {$x_i\gets 0$}\EndIf
  \State{\textbf{goto} line~\ref{l2}}
  \EndIf
  \EndFor
  \If{$x_i = 0$} {\textbf{goto} line~\ref{l1}} \EndIf
  \For{$j\gets i\oplus 1$ \textbf{cyclically to} $f$}
  \If{$x_j=1$} {\textbf{goto} line \ref{l2}}\EndIf
  \EndFor
  \State{\textbf{critical section}}
  \State{$z_i\gets 1-z_i$}
  \State{$x_i\gets 0$}
  \State{$y_i\gets 0$}
 \end{algorithmic}
\end{algorithm}

We find that the algorithm indeed satisfies mutual exclusion and reachability of the critical section when using non-atomic registers.
However, in the process of modelling it became clear these results are only valid when the computation on line \ref{lamport-3bit-functioncall} is implemented in a very particular manner, something which is not emphasised in the algorithm's presentation.
In our first model, we handled this line as follows: we went through the constructed ordered list $\gamma$ from the smallest to the largest element. For each, we read the associated $z$-value as well as the $z$-value of the element it has to be compared against according to the $\mathit{CGV}$ function. 
As soon as we found one which satisfies the equality, that value was chosen for $f$.
This seems to match the description of the algorithm from \cite{Lamport86Mutex2}, ``evaluating [$z$] at $j$ requires a read of the variable $z_j$.''
Yet, this implementation leads to a violation of the reachability of the critical section, even when atomic registers are used.

The violation is caused when one thread is updating its $z$-value in the exit protocol, while a different thread has to read this $z$-value multiple times for the computation on line \ref{lamport-3bit-functioncall}.
Consider the following situation with two threads and atomic registers:
\begin{enumerate}
    \item Thread 1 executes the entry protocol successfully because it is the only competing thread. It gets to its exit protocol and sets $z_1$ to $1$. It then completes the rest of the exit protocol.
    \item Thread 1 starts the competition again. Once again, it is the only competing thread and hence can reach the critical section and start the exit protocol. It has not yet updated $z_1$ to be $0$ again.
    \item Thread 0 starts the competition. At line \ref{l2} it finds that $y_0 = y_1 = 1$ so $\gamma = [0, 1]$.
    \item On line \ref{lamport-3bit-functioncall}, thread 0 tests if $\mathit{CG}(z,\gamma,0) = 1$. For this, it needs to compare $z_0$ and $z_1$. It finds $z_0 = 0$ and $z_1 = 1$. We find $\mathit{CG}(z, \gamma, 0) = 0$, so $0$ is not deemed a valid value for $f$.
    \item Thread 1 now updates $z_1$ to be $0$.
    \item Next, thread 0 checks $\mathit{CG}(z, \gamma, 1)$. It once again reads $z_0$ and $z_1$, but now it wants those to have different values. It reads $z_0 = 0$ and $z_1 = 0$. Hence, $\mathit{CG}(z, \gamma, 1) = 0$ so 1 is not deemed a valid value for $f$.
\end{enumerate}
The algorithm does not account for there being no valid value for $f$, because this situation is not meant to occur. 
As a consequence, in our model thread 0 simply cannot take actions anymore once this situation is reached and will therefore never reach the critical section again.
In \cite{Lamport86Mutex2} a lemma is proven which states that there is always some $i$ for which $\mathit{CG}(v, \gamma, i) = 1$. 
But that proof does not account for a value changing while the comparisons are being done.

This situation can be avoided by reading all the $z$ values once before starting line \ref{lamport-3bit-functioncall}, and storing the result locally. 
Indeed, once we modelled the algorithm using this approach we found the results shown in \autoref{tab:verifres}.
This once again highlights the importance of minor implementation details when dealing with non-atomic registers.

\subsection{Szymanski (flag)}\label{app:szymanski-flag-mutex}
As stated in \autoref{subsec:szymanski-flag-mutex}, Szymanski does not claim the flag algorithm is valid when using non-atomic registers.
He instead claims that when the $\mathit{flag}$ register is implemented as three bits as shown in \autoref{tab:szy-flag2bits}, the algorithm is resistant to regular registers.
Pseudocode is given in \cite{Szy88} that incorporates this change, as well as other changes to make the algorithm resistant to thread failure and restarts.
However, some formatting issues in the pseudocode presentation means we are not certain exactly what was intended there.
Instead of modelling that pseudocode, we therefore made the translation from the flag-algorithm to a three bit implementation ourselves and modelled that.
The result is shown in \autoref{alg:szy-flag-bits}.

\begin{table}
\centering
\begin{tabular}{|c|c|c|c|}
\hline
\textit{flag} & \textit{intent} & \textit{door\_in} & \textit{door\_out} \\ \hline
0             & 0               & 0                 & 0                  \\ \hline
1             & 1               & 0                 & 0                  \\ \hline
2             & 0               & 1                 & 0                  \\ \hline
3             & 1               & 1                 & 0                  \\ \hline
4             & 1               & 1                 & 1                  \\ \hline
\end{tabular}
\caption{Translating the integer register $\mathit{flag}$ to three Boolean registers $\mathit{intent}$, $\mathit{door\_in}$ and $\mathit{door\_out}$.}
\label{tab:szy-flag2bits}
\end{table}

\begin{algorithm}
\caption{Szymanski's flag algorithm implemented with bits}\label{alg:szy-flag-bits}
\begin{algorithmic}[1]
  \State{$\mathit{intent}[i]\gets 1$}\label{szy-bits-1}
  \State{\textbf{await} $\forall j.\ \mathit{intent}[j] = 0 \lor \mathit{door\_in}[j] = 0$}\label{szy-bits-2}
  \State{$\mathit{door\_in}[i]\gets 1$}\label{szy-bits-3}
  \If{$\exists j.\ \mathit{intent}[j] = 1 \land \mathit{door\_in}[j] = 0$}\label{szy-bits-4}
    \State{$\mathit{intent}[i]\gets 0$}\label{szy-bits-5}
    \State{\textbf{await} $\exists j.\ \mathit{door\_out}[j]=1$}\label{szy-bits-6}
  \EndIf
  \If{$\mathit{intent}[i] = 0$}\label{szy-bits-7}
    \State{$\mathit{intent}[i] \gets 1$}\label{szy-bits-8}
  \EndIf
  \State{$\mathit{door\_out}[i]\gets 1$}\label{szy-bits-9}
  \State{\textbf{await} $\forall j < i.\ \mathit{door\_in}[j] = 0$}\label{szy-bits-10}
  \State{\textbf{critical section}}\label{szy-bits-11}
  \State{\textbf{await} $\forall j>i.\ \mathit{door\_in}[j]=0 \vee
    \mathit{door\_out}[j]=1$}\label{szy-bits-12}
  \State{$\mathit{intent}[i]\gets 0$}\label{szy-bits-13}
  \State{$\mathit{door\_in}[i]\gets 0$}\label{szy-bits-14}
  \State{$\mathit{door\_out}[i]\gets 0$}\label{szy-bits-15}
\end{algorithmic}
\label{alg:szymanski-flag-bits}
\end{algorithm}

Interestingly enough, we find that mutual exclusion no longer holds even with atomic registers when this change is made.
With two threads, mCRL2 reports that mutual exclusion holds with atomic registers (although not with safe or regular). 
For three threads, the following counterexample is found:
\begin{enumerate}
    \item Thread 2 runs through the entire algorithm until line \ref{szy-bits-14}. At this point, $\mathit{intent}[2] = 0, \mathit{door\_in}[2] = 1$ and $\mathit{door\_out}[2] = 1$.
    \item Threads 0 and 1 can both get past line \ref{szy-bits-2}, since $\mathit{door\_in}[0] = \mathit{door\_in}[1] = 0$ and $\mathit{intent}[2] = 0$.
    \item Thread 1 continues further, at line \ref{szy-bits-4} it sees $\mathit{intent}[0] = 1$ and $\mathit{door\_in}[0] = 0$ so it has to execute lines \ref{szy-bits-5} and \ref{szy-bits-6}. It can immediately escape line \ref{szy-bits-6} however, because $\mathit{door\_out}[2] = 1$. 
    \item Thread 1 continues through lines \ref{szy-bits-7}, \ref{szy-bits-8} and \ref{szy-bits-9}. On line \ref{szy-bits-10}, it sees $\mathit{door\_in}[0] = 0$ so it can enter the critical section.
    \item Thread 0 continues on line \ref{szy-bits-3}. There is no thread with $\mathit{intent}$ set to true and $\mathit{door\_in}$ to false, so it directly gets to line \ref{szy-bits-9} and \ref{szy-bits-10}. Since it has the lowest thread id, it can immediately enter the critical section.
\end{enumerate}
This counterexample relies on the resetting of the variables in the exit protocol happening separately, rather than all at once.
The properties can be made true if the order of the resets in the exit protocol is changed.
If the order is $\mathit{door\_out}, \mathit{door\_in}, \mathit{intent}$ then both properties hold with three threads and atomic registers.
This is not a desirable solution however. While we did not analyse starvation freedom formally, reconsidering the above counterexample makes it easy to observe that the following scenario is possible if $\mathit{door\_out}$ and $\mathit{door\_in}$ are reset before $\mathit{intent}$ is:
\begin{enumerate}
    \item Thread 2 runs through the entire algorithm until it gets to the exit protocol. Here it resets $\mathit{door\_out}$ and $\mathit{door\_in}$ but not $\mathit{intent}$ yet.
    \item Threads 0 and 1 both get past line \ref{szy-bits-2} because $\mathit{door\_in}[0] = \mathit{door\_in}[1] = \mathit{door\_in}[2] = 0$. On line \ref{szy-bits-4}, both see $\mathit{intent}[2] = 1$ and $\mathit{door\_in}[2] = 0$, meaning they both go to lines \ref{szy-bits-5} and \ref{szy-bits-6}.
    \item If thread 2 never chooses to re-attempt access of the critical section, then $\mathit{door\_out}[2]$ will never become true again. And so threads 0 and 1 will never escape line \ref{szy-bits-6}. 
\end{enumerate}
So while reachability of the critical section holds, it relies on thread 2 always wanting to re-enter the competition.
We foresee no such issues if the order of resets is $\mathit{door\_out}$, $\mathit{intent}$, $\mathit{door\_in}$.
Although further formal verification is needed to confirm this reset order has no complications.

Note that regardless of the exit protocol's exact implementation, the algorithm still does not ensure mutual exclusion with safe or regular registers.
For example, with two threads and safe registers, mCRL2 generates the following counterexample for mutual exclusion:
\begin{enumerate}
    \item Both threads 0 and 1 get through lines \ref{szy-bits-1} and \ref{szy-bits-2}.
    \item Thread 0 starts writing 1 to $\mathit{door\_in}[0]$, but does not finish this write yet.
    \item Thread 1 continues through line \ref{szy-bits-3}. On line \ref{szy-bits-4}, it reads $\mathit{door\_in}[0] = 1$ with overlap. It also sees $\mathit{door\_in}[1] = 1$ so can continue to line \ref{szy-bits-7}.
    \item Thread 1 can simply continue through lines \ref{szy-bits-7} and \ref{szy-bits-9}. On line \ref{szy-bits-10}, it reads $\mathit{door\_in}[0] = 0$ with overlap, and can therefore enter the critical section.
    \item Thread 0 finishes the write on line \ref{szy-bits-3}. On line \ref{szy-bits-4} it sees $\mathit{door\_in}[0] = \mathit{door\_in}[1] = 1$ so it continues on line \ref{szy-bits-7}. It does \ref{szy-bits-7}, \ref{szy-bits-9} and then on line \ref{szy-bits-10} it does not need to check any other thread's $\mathit{door\_in}$ values because it has the lowest thread id. Thread 0 can enter the critical section.
\end{enumerate}
This counterexample is also valid for SWMR regular registers, since those allow two reads overlapping the same write to first read the new and then the old value.

\subsection{Szymanski (3-bit linear wait)}\label{app:szymanski-3bit-mutex}
An alternative version of Szymanski's algorithm is the 3-bit linear wait algorithm from \cite{Szy90}.
The pseudocode is presented in \autoref{alg:szymanski-3bit}.

  \begin{algorithm}
    \caption{Szymanski's 3-bit linear wait algorithm}
    \label{alg:szymanski-3bit}
    \begin{algorithmic}[1]
      \State{$a_i\gets 1$}\label{szy-3bit-1}
      \lFor{$j\gets 0$\ \textbf{to}\ $N{-}1$} {\textbf{await} $s_j = 0$}\label{szy-3bit-2}
      \State{$w_i\gets 1$}\label{szy-3bit-3}
      \State{$a_i\gets 0$}\label{szy-3bit-4}
      \While{$s_i=0$}\label{szy-3bit-5}
\State{$j \gets 0$}\label{szy-3bit-6}
        \lWhile{$j < N\wedge a_j = 0$}{$j\gets j+1$}\label{szy-3bit-7}
          \If{$j=N$}\label{szy-3bit-8}
            \State{$s_i \gets 1$}\label{szy-3bit-9}
            \State{$j\gets 0$}         \label{szy-3bit-10}
            \lWhile{$j<N\wedge a_j=0$}{$j\gets j+1$}\label{szy-3bit-11}
            \If{$j< N$} {$s_i\gets 0$}\label{szy-3bit-12}
           \Else\label{szy-3bit-13}
             \State{$w_i\gets 0$}\label{szy-3bit-14}
             \lFor{$j\gets 0$\ \textbf{to} $N-1$}{\textbf{await} $w_j=0$}\label{szy-3bit-15}
           \EndIf
         \EndIf
         \If{$j<N$}\label{szy-3bit-16}
           \State{$j\gets 0$}\label{szy-3bit-17}
           \lWhile{$j<N \wedge (w_j=1\vee s_j=0)$}{$j\gets j+1$}\label{szy-3bit-18}
         \EndIf
         \If{$j \neq i \wedge j<N$}\label{szy-3bit-19}
         \State{$s_i\gets 1$}\label{szy-3bit-20}
         \State{$w_i\gets 0$}\label{szy-3bit-21}
         \EndIf
     \EndWhile
      \lFor{$j\gets 0$ \textbf{to} $i-1$}{\textbf{await} $s_j = 0$}\label{szy-3bit-22}
      \State{\textbf{critical section}}\label{szy-3bit-23}
      \State{$s_i\gets 0$}\label{szy-3bit-24}
     \end{algorithmic}
  \end{algorithm}

Surprisingly, our verification shows this algorithm does not ensure mutual exclusion even when using atomic registers. 
The counterexample once again requires at least three threads. 
The first counterexample generated by mCRL2 relies on reading $w_j$ and $s_j$ separately on line \ref{szy-3bit-18}. 
This was likely unintended behaviour, but as far as we can tell it is not excluded in the algorithm's description in \cite{Szy90}; not to mention that enforcing this would require a semaphore which returns us to the issue of assuming a lower-level solution to the mutual exclusion problem.

For completeness, we still model a variant where a semaphore is used to protect every write to a $w$- or $s$-register, as well as the two reads on line \ref{szy-3bit-18}, but nothing else.
This leads to the following counterexample with three threads and atomic registers:
\begin{enumerate}
    \item Threads 0, 1 and 2 all execute lines \ref{szy-3bit-1} and \ref{szy-3bit-2}. We now have that all $a$'s are 1, all $w$'s 0 and all $s$'s 0.
    \item Thread 1 continues the competition. It executes lines \ref{szy-3bit-3} and \ref{szy-3bit-4}, on line \ref{szy-3bit-5} it sees $s_1 = 0$ so it goes into the while-loop. On line \ref{szy-3bit-7} it sees $a_0 =1 $ so it breaks out with $j = 0$ which means the condition on line \ref{szy-3bit-8} evaluates to false and the condition on line \ref{szy-3bit-16} evaluates to true. 
    \item Thread 0 does the same. On line \ref{szy-3bit-7}, it also breaks early because $a_2 = 1$ hence it also ends up in the body of the if-statement on line \ref{szy-3bit-16}.
    \item Thread 0 reads $w_0 = 1, s_0 = 0$ on line \ref{szy-3bit-18}; it continues with $j = 1$.
    \item Thread 1 reads $w_0 = 1, s_0 = 0, w_1 = 1, s_1 = 0$ on line \ref{szy-3bit-18}, so it continues with $j = 2$.
    \item Thread 2 now continues. It executes lines \ref{szy-3bit-3} to \ref{szy-3bit-6}. On line \ref{szy-3bit-7}, it finds that all $a$-values are $0$, so it continues on line \ref{szy-3bit-9}. Here it sets $s_2$ to 1. On line \ref{szy-3bit-11} it once again sees that all $a$-values are $0$, so it continues on line \ref{szy-3bit-14} where it sets $w_2$ to $0$.
    \item Thread 1 reads $w_2 = 0, s_2 = 1$, hence it breaks out of the loop on line \ref{szy-3bit-18} with $j = 2$. The condition on line \ref{szy-3bit-19} evaluates to true so thread 1 executes both $s_1 \gets 1$ and $w_1 \gets 0$.
    \item Thread 1 goes back to line \ref{szy-3bit-5}. Here it sees $s_1 = 1$, so it goes to line \ref{szy-3bit-22}. It reads $s_0 = 0$ and can therefore enter the critical section.
    \item Thread 0 now reads $w_1 = 0, s_1 = 1$, so it breaks out of the loop on line \ref{szy-3bit-18} with $j = 1$. The condition on line \ref{szy-3bit-19} evaluates to true so it executes $s_0 \gets 1$ and $w_0 \gets 0$. 
    \item Thread 0 goes back to line \ref{szy-3bit-5} where it sees $s_0 = 1$. It then goes to line \ref{szy-3bit-22} where it does not need to check the $s$-value of any thread since it has the lowest id. Thread 0 can enter the critical section.
\end{enumerate}

The issue is that the third thread allows the other two to satisfy exists-conditions when this should not be happening.
This same counterexample is of course valid with regular and safe registers with three threads.
We did find that with only two threads, the semaphore does ensure the two properties even when safe registers are used.
\end{document}